\documentclass{extarticle}
\usepackage{amsmath, amsthm, amssymb}  
\usepackage[skip=0pt,
justification=raggedright,singlelinecheck=false]{caption}

\usepackage{geometry,verbatim}
\usepackage{graphicx}
\usepackage{longtable}
\usepackage{threeparttable}
\usepackage{booktabs}
\usepackage{amsmath}
\usepackage{amsfonts}
\usepackage{amssymb}
\usepackage{tikz}
\usepackage{float}
\usepackage{setspace}
\usepackage{xr}
\usetikzlibrary{snakes,calc}
\usetikzlibrary{patterns}
\usepackage{dcolumn}
\newcolumntype{d}[1]{D{.}{.}{#1}}
\usepackage{multirow}

\usepackage[labelformat=simple]{subcaption}

\usepackage{pdfpages}

\usepackage[utf8]{inputenc}
\usepackage[english]{babel}
\usepackage{csquotes}

\usepackage[backend=biber,style=authoryear-comp,
maxcitenames=2,sorting=nyt,]{biblatex}

\usepackage{authblk}

\newtheorem{theorem}{Theorem}

\newtheorem{proposition}[theorem]{Proposition}

\usepackage{hyperref}

\newcommand{\Figtext}[1]{%
	\begin{tablenotes}[para,online, flushleft]
		\footnotesize
		\hspace{-0.25cm}
		#1
	\end{tablenotes}
}

\newcommand{\Fignote}[1]{\Figtext{~#1}}

\usepackage{datetime}
\newdateformat{monthyeardate}{%
	\monthname[\THEMONTH], \THEYEAR}

\graphicspath{ {images/} }
\begin{document}

\title{Can political gridlock undermine checks and balances? A lab experiment.} 	

	\author{Alvaro Forteza \thanks{%
			Departamento de Econom\'{\i}a, FCS-UDELAR, Uruguay.
			Alvaro.forteza@cienciassociales.edu.uy.} \ 
		and Irene Mussio \thanks{Newcastle University Business School, Irene.mussio@cienciassociales.edu.uy, Departamento de Econom\'{\i}a, FCS-UDELAR, Uruguay } \
					and Juan S. Pereyra \thanks{Department of Economics, Universidad de Montevideo, Uruguay.
			jspereyra@um.edu.uy.}
		\thanks{We are grateful with Marina Agranov, Juan Camilo Cárdenas, Nathan Canen, María Paz Espinosa, Germán Feierherd, Leopoldo Fergusson, Santiago Lopez Cariboni, Daniel Monte, and Graciela Sanroman for their insightful comments and suggestions. We are also in debt with participants at the Central Bank of Uruguay, Department of Economics - PUC Chile, the Latin American Network in Economic History and Political Economy (LANE HOPE), and at the Seminario Anual - dECON.
			Financial support from ANII, FCE\_1\_2017\_1\_135851, is gratefully acknowledged. This research was approved by the Ethics Committee of the Faculty of Medicine of Universidad de la República and the Ethics Committee of McMaster University (3699). It was pre-registered at the Wharton Credibility Lab (\href{https://aspredicted.org/blind.php?x=LUM_BWT}{\#26622}). The codes and database are available at \href{https://github.com/alforteza/weakening-checks-and-balances/tree/main}{Git Hub repository}. Luciana Cantera, Manuel Adler and Romina Quagliotti provided excellent assistance in the experimental sessions. The usual disclaimer applies.}}
	\date{\monthyeardate\today}
	\maketitle
	
	\baselineskip18pt
	
	\thispagestyle{empty}
	
	\setcounter{page}{0}

	\begin{center}
		\textbf{Abstract}
	\end{center}

If checks and balances are aimed at protecting citizens from the government's abuse of power, why do they sometimes weaken them? We address this question in a laboratory experiment in which subjects choose between two decision rules: with and without checks and balances. Voters may prefer an unchecked executive if that enables a reform that, otherwise, is blocked by the legislature. Consistent with our predictions, we find that subjects are more likely to weaken checks and balances when there is political gridlock. However, subjects weaken the controls not only when the reform is beneficial but also when it is harmful. 


	
	\pagebreak

\begin{flushright}\textit{
		``In every region of the world,\\changing times have boosted public demand\\ for more muscular, assertive leadership. (...) \\ We’re now in the strongman era.'' \\ Ian Bremmer TIME, May 2018.}
\end{flushright}
\medskip

	\section{Introduction}
	
Episodes of weakening of checks and balances with citizens support represent a puzzle and a challenge for political economy. Although the general trend in last decades has been an increase in checks and balances, some developing democracies have experienced reverse trends \parencite{besley2011pillars}. This was the case of several Latin American countries during the 1990s and 2000s, Recep Erdo\u{g}an in Turkey, Viktor Orbán in Hungary, and Vladimir Putin in Russia \parencite{Acemoglu2013c, Forteza2019}. If checks and balances are aimed at protecting citizens from governments abuse of power, why do citizens sometimes decide to remove them? In this paper, we build a model that explains citizens' willingness to undermine controls on the executive, and present the results of a laboratory experiment to provide supportive evidence to our hypotheses.

We explore the idea that voters may support the loosening of checks and balances as a reaction to \textbf{political gridlock}. That is, to a situation where the executive proposes a reform and the legislature proposes to keep the status quo. In this environment, the basic trade-off voters face is between effective control of the executive and reform. Loosening controls facilitates the reform proposed by the executive but at the cost of increased corruption. Keeping the controls puts a break on corruption but at the cost of no reform.

We extend this argument with a model where voters have to decide between two institutional settings: \textbf{checks and balances} ($CB$), and \textbf{special powers} ($SP$). All voters have the same preferences over policies: they prefer the one that matches the state of nature. However, they do not observe the state of nature: they only know the probability with which a reform is beneficial. The executive and the legislature observe the state of nature and then propose and commit to a policy.\footnote{We assume politicians credibly propose policies before voters decide over special powers. This timing is meant to capture real life episodes in which rulers demanded special powers to advance a reform program. In these episodes there was little doubt that the announced policies would be implemented, so there were no commitment issues. We represent this commitment ability assuming it stems from politicians types: some politicians prefer the status quo, some the reform and some the matching of the policy with the state of nature.} Each politician can be conservative (always proposes the status quo), reformist (always proposes the reform), or unbiased (proposes the policy to match the state of nature). Voters observe the type of the politicians in each branch and their proposals, and decide between $CB$ and $SP$. Under $CB$, a reform is implemented if, and only if, both branches agree, otherwise the status quo remains. With $SP$, the policy proposed by the executive is always implemented. Additionally, the extraction of rents by the executive, which entails a cost for voters, takes place only under $SP$.

With the model in hand, we ask the question: does political gridlock increase the probability of $SP$? 
The answer depends on the nature of the political gridlock. On the one hand, political gridlock \textbf{raises} the probability of $SP$ if either (i) the executive is reformist, the legislature is conservative and the reform is ex ante beneficial ---this is our hypothesis 1 (H1)---, or (ii) the executive is unbiased and proposes reform, and the legislature is conservative (H2). On the other hand, political gridlock \textbf{reduces} the probability of $SP$ if either (i) the executive is reformist, the legislature is conservative and the reform is ex ante harmful (H3), or (ii) the executive is reformist and the legislature is unbiased and proposes the status quo policy (H4).

Existing literature provides limited empirical evidence that political gridlock causes special powers. Using case studies, \textcite{Forteza2019} argue that several Latin American presidents weakened $CB$ with strong popular support alleging that there was a political gridlock that impeded reform. The president of Peru, Alberto Fujimori, led an auto-coup in 1992. He argued that  ‘the lack of identiﬁcation of some fundamental institutions with the national interests, like the legislature and the judiciary, blocks the actions of the government oriented to national reconstruction and development’. Public opinion polls document the strong support he received (71\% of respondents supported the closing of the legislature). In Argentina in the 90s, a very popular president, Carlos Menem, weakened the judiciary with similar arguments.\footnote{According to Menem's minister of justice, ‘An administration . . . cannot govern with a judiciary whose views are antagonistic to those of the government. If the Court were to have a vision completely diﬀerent from ours and to declare our laws unconstitutional, we could not implement our political and economic plans’ \parencite{chavez2004rule}.} Hugo Chávez, president of Venezuela between 1999 and 2013, had an ambitious program of reforms. With strong popular support, he strengthened the executive power vis-à-vis the legislature and the judiciary. The weakening of $CB$ was instrumental to advance his radical program. Rafael Correa, president of Ecuador between 2007 and 2017, also alleged that the strengthening of executive powers was necessary to advance reforms. Evo Morales, president of Bolivia between 2006 and 2019, also enjoyed strong popular support. In 2009 he obtained 61\% of popular votes in a plebiscite to reform the Constitution. The reform included a significant strengthening of the powers of the Executive. These case studies are illustrative, but do not provide causal evidence. We thus appeal to a laboratory experiment.\footnote{In all of these cases, public opinion support was instrumental for the executive to obtain special powers. There have been, however, other cases where citizens were divided in support of the reforms (as is currently the case with the judicial reform in Israel). We are not considering these cases as we want to understand why citizens support the dismantling of checks and balances.}

As voters in the model, participants of the experiment do \textbf{not} observe the state of nature, they only know the probability with which the reform is beneficial. They make decisions in 14 scenarios, which correspond to all possible combinations of the executive and legislature types and proposals. In each of them, subjects are informed of the policy proposals and politicians' types. They have to (i) determine whether there is political gridlock and the policy implemented with each rule, and (ii) choose between $CB$ and $SP$. We are especially interested in a scenario where the executive is reformist and the legislature is conservative. In this case, a political gridlock arises, and subjects are only aware that the reform is beneficial with a certain probability. Under checks and balances the executive cannot extract rents. In this regard, subjects are protected. However, this decision rule does not lead to any policy change. If subjects wish to implement a different policy, they must incur the costs associated with special powers. This is how we interpret the experiment in terms of our model.\footnote{In all but one of the treatments, the framing is neutral: the executive and the legislature were introduced as two decision-makers, checks and balances as rule 1, special powers as rule 2, and rents were explained as the cost of adopting rule 2. Nevertheless, we observe similar results in the treatment where we frame everything in terms of political decisions in a presidential system.}

We find evidence supporting some but not all of our hypotheses. \textbf{Subjects are more likely to support the weakening of CB with than without political gridlock}. This occurs in the four scenarios of political gridlock, which grants support to H1 and H2 but not to H3 and H4. As expected, we find that political gridlock caused by biased politicians when the reform was ex ante beneficial raises the probability of $SP$ (H1). Also, political gridlock with an unbiased executive raises the frequency of $SP$ (H2). However, contradicting H3, subjects also grant $SP$ in a higher proportion when faced with gridlock caused by biased politicians when the reform is ex ante harmful. Also, subjects grant $SP$ more frequently when faced with gridlock where the legislature is unbiased and truthfully ``warns'' that the reform is not beneficial, contradicting H4. 

Overall, the results of the experiment show an excess of support for $SP$. Participants undermine the controls on the executive even when this entails a loss in terms of their payoffs. We explore some possible factors driving this result. First, we explore if the excess of $SP$ was driven by individuals who made more mistakes. We do not find statistically significant differences in the impact of the political gridlock associated with harmful reforms on the frequency of $SP$, chosen by subjects who did and did not make mistakes.\footnote{In our setting, reforms can be ex-ante and ex-post harmful. Reforms are ex-ante harmful when both the executive and the legislature are biased and the probability that the reform is beneficial is low (this case corresponds to our hypothesis H3). Reforms are ex-post harmful when the executive is biased, the legislature is unbiased and warns that the reform is not beneficial (this case corresponds to our hypothesis H4).} Then, we replicate the analysis for the different variables and framing treatments: risk aversion, gender, political affiliation, support for strong leaders, and corruption and political framing. We do not find statistically significant differences in the excess of support for $SP$ for these variables. However, as we discuss below, there could be power issues regarding these additional tests. 

Finally, a motive for concern could be that our experiment induces the excess of SP by design. In each of the 14 scenarios where participants take decisions, we ask them if there is political gridlock. So we deliberately induced subjects to focus on it. However, we do it in a very mild way, much milder than what citizens in the real world are usually subject to. Indeed, executives suffering the blocking of their program by a legislature typically explicitly ask for special powers, and forcefully argue that the opposition does not allow them to do what has to be done \parencite[see][for some vivid narratives of this type of claims in the case of several strongmen in Latin America]{Forteza2019, Acemoglu2013c}. In this regard, while our experiment does not directly explore the effects of rulers' speeches on subjects willingness to support the weakening of checks and balances, it does show that subjects can easily be induced to do it by simply prompting them to think about political gridlocks.  

Our paper is related to several strands of the literature. Citizens in our model and subjects in our experiment face a basic trade-off between political control and governmental delegation that is at the center of the extensive literature on separation of powers and checks and balances \parencite[see, among many others,][]{Persson1997, Persson2000, Stephenson2010, Buisseret2016, Fox2010, ODonnell1998, Besley2018}.\footnote{Classic writers also emphasized the role of checks and balances as a protection  of minorities from ``the tyranny of the majority'' \parencite{Montesquieu1748, Hamilton2009, Locke1689}. For recent papers that also discuss this theme see \textcite{Aghion2004, Hayek1960, Buchanan1962, Buchanan1975a, Maskin2004}.} 
The main potential drawback of checks and balances in our environment is the impossibility of policy change caused by the legislative veto on executive proposals. In this regard, our paper also builds on the literature on ``veto players'' initiated by \citeauthor{Tsebelis1995} \parencite[see][among many others]{Tsebelis1995, Aghion2004, Tommasi2014}. 

Our paper is also closely related to the recent literature aimed at explaining the weakening of democracy \parencite{Svolik2020} and, more specifically, the weakening of checks and balances with voters support \parencite{Acemoglu2013c, Forteza2019, Ryvkin2017}. \citeauthor{Acemoglu2013c} argue that the poor majority supports the dismantling of CB because politicians are less tempted to accept bribes from the rich elites if they can extract rents than if they cannot. Therefore, their story explains episodes in which voters support the weakening of CB under executives that favor the poor. But, in our view, it does not explain so well why equally popular politicians holding a pro-market agenda got the support of citizens to undermine these controls when the reform agenda was firmly supported by the elites. Our model can explain the loosening of CB that took place during both pro-market and anti-market reforms. \textcite{Ryvkin2017} test a simplified version of \textcite{Acemoglu2013c} in a laboratory experiment. As in our design, subjects can shut down democracy through majority voting. They focus on the effect of inequality and productivity on the likelihood of democracy breakdown. Our paper complements \textcite{Ryvkin2017} by analyzing the effect of political gridlock.

Finally, we owe much to the literature on experimental political economy. One of the main reasons to use a laboratory experiment as a tool to understand citizens' decisions on the strengths of checks and balances lies in the difficulties of gathering observational data that allows for identifying causal relationships. Even when some data on checks on the executive is available for a set of countries (see, for example, the Polity IV index of executive constraints), it is more difficult to have reliable data on political gridlock, and, more importantly, it is very difficult to identify exogenous variation to test causal relationships. 

The rest of the paper proceeds as follows. In Section \ref{model} we present the model and the theoretical predictions. Section \ref{design} describes the experimental design. The results of the lab experiment are presented in Section \ref{results}. The paper ends with a few concluding remarks in Section \ref{concluding}. An online appendix contains the proofs of propositions, some descriptive statistics, balance tests, alternative Fisher tests, the experimental instructions, risk aversion measurement and the post-experimental questionnaire.
	
	\section{The model and theoretical predictions}\label{model}
	
In this section, we first introduce the model and then the main predictions to be tested.
	
	\subsection{The model}

We use a probabilistic voting model to study voters' decision on executive special powers.
	
	\paragraph{States of Nature} There are two possible states of nature: $s=0$ and $s=1$. The a priori probability of $s=1$ is $q\in [0,1]$. 
	
	\paragraph{Government.}There is a government composed of two branches, the executive ($X$) and the legislature ($L$). $X$ and $L$ observe the state of nature $s \in \{0,1\}$, and make policy proposals $p_X, p_L \in \{0,1\}$. For concreteness, we assume that the status quo policy is $p_0=0$. There are three types of politicians: (i) ``conservative'' ($X_C, L_C$) who always propose the status quo policy $p^{}_i = 0, i \in \{X,L\}$, (ii) ``reformist'' ($X_R, L_R$) who always propose the reform $p^{}_i = 1, i \in \{X,L\}$, and (iii) ``unbiased'' ($X_U, L_U$) who match the state of nature $p^{}_i = s, \ i \in \{X,L\}$. Note that neither $X$ nor $L$ are strategic agents.\footnote{\textcite{Forteza2019} present a model on similar lines and with strategic executive and legislature. The qualitative results regarding voters behavior are not different from what we find in the present simpler setting. } 	
	
	\paragraph{Voters.}	Voters observe the politicians' types and their proposals. They \textbf{do not} observe the state of nature, they know $q$. Once they observe the type of each politician, and their proposals, voters choose one among two possible institutions that lead to different mappings from proposals to implemented policies: checks and balances ($CB$) and special powers ($SP$). 
	
	\paragraph{Institutional arrangements.} With $CB$ the implemented policy is equal to both branches proposals if there is agreement, and to the status quo policy otherwise. With $SP$, the will of the executive prevails. Equations (\ref{eq:policyrules}) and Table \ref{table:proposalstopolicies} summarize how these two institutions work regarding policy $p$.
	
	\begin{align} \begin{array}{ll}
	p(CB) =  \left\{	\begin{array}{lc}  	p_X   	& if \ p_X = p_L \\
	p_{0}=0  	& if \ p_X \neq p_L \end{array} \right.    \\
	p(SP) = p_X
	\end{array} \label{eq:policyrules}
	\end{align}
	
	\begin{table}[H]
		\begin{center}
			\captionsetup{justification=centering}
			\caption{From proposals to policies with CB and SP.}\label{table:proposalstopolicies}
			\begin{tabular}{|c|c|c|c|} \hline
				
				$p_X$ & $p_L$ & $p(CB)$ & $p(SP)$ \\  \hline
				
				0 & 0 & 0 & 0 \\ \hline
				0 & 1 & 0 & 0 \\ \hline
				1 & 0 & 0 & 1 \\ \hline		
				1 & 1 & 1 & 1 \\ \hline 	
			\end{tabular}
		\end{center}
		\end{table}
	
	\paragraph{Rents.}With $CB$, there is an effective control of the government and hence corruption does not arise ($r=0$). With $SP$, the government extracts an amount of rents $r>0$.

	\paragraph{Preferences} Citizens care about policies and rents. They prefer the policy that matches the state of nature, $p=s$, and no rent extraction. Their utility function is:
	
	\begin{equation} \label{eq:utility}
	v(p,r)=  -a\mathbb{E}_s\left[(p - s)^2 \right] - r  ,
	\end{equation}
	where the parameter $a\geq 0$ captures the relative weight citizens give to policy mismatch and rents.
	
	Voter $i$ expected utility gains from voting for $CB$ rather than $SP$ is:
	\[ v(CB) - v(SP) +\varepsilon_{i} \] 
	where $v(CB)= v(p(CB),r(CB))$, $v(SP)= v(p(SP),r(SP))$, $r(CB)=0$, $r(SP)=r$, and $\varepsilon_{i}\in [-\infty,+\infty]$ is a random variable with mean $0$, and distribution function $F$.\footnote{$\varepsilon$ captures in a very stylized form voters' heterogeneity as well as the uncertainty that the analyst has about relevant citizens' traits, such as their preferences, attitudes, etc. The assumption is used in the tradition of random utility  \parencite{McFadden1975} and probabilistic voting models \parencite{Lindbeck1987}. For a similar assumption in the context of a lab experiment on electoral accountability see \textcite{Landa2015}.}

	The timing is as follows. First, $X$ and $L$ propose policies $p_X$ and $p_L$. Second, voters decide whether or not to give $X$ special powers. At this time, voters observe (i) $X$ and $L$ types, (ii) policy proposals of $X$ and $L$, and (iii) the realization of their preference shock $\varepsilon_{i}$, but they \textbf{do not observe the state of nature}. 
	
	\subsection{Predictions to be tested}
		
	Using the policy rules (\ref{eq:policyrules}) in the utility function (\ref{eq:utility}) we get:	
	\begin{align} \begin{array}{ll}
	v(CB)=  \left\{  \begin{array}{lc}   -a\mathbb{E}_s\left[(p_X - s)^2 \right]     & if \ p_X=p_L, \\
	-a\mathbb{E}_s\left[(p_0 - s)^2 \right]    & otherwise \end{array} \right. \\
	v(SP)=  -a\mathbb{E}_s\left[(p_X - s)^2 \right] -r   
	\end{array} \label{eq:vCB&vSP}
	\end{align}	
	
	Note that $v(SP) - v(CB)=-r$ if either  $p_X=p_L$ or $p_X=0$ and $p_L=1$. The only remaining case is $p_X=1$ and $p_L=0$, that is, $X$ proposes a policy change blocked by $L$, which we refer as \textbf{political gridlock}. 
	If neither $X$ nor $L$ is unbiased, then voters have to rely on the prior probability $q$ to compute the expected utility in the above equations. In this case we have:  $v(SP) - v(CB)=-a[\mathbb{E}_s(1-s)^2-\mathbb{E}_s(s^2)]-r=-a(1-2q)-r$.
	If at least one of the two rulers is unbiased, then voters can use Bayes rule to deduce the true state of nature from their policy proposals. Therefore, if $X$ is unbiased $v(SP) - v(CB)=-r+a$, and if $L$ is unbiased $v(SP) - v(CB)=-r-a$.
	
	Using the previous observation and equations (\ref{eq:vCB&vSP}) we have that:
 
	\begin{align} \label{eq:vSP-vCB}
	v(SP) - v(CB)  =  \left\{  \begin{array}{ll} 	- r-a(1-2q) 	& if \ \ X_R, L_C,  \\
	- r-a 			& if \ \  X_R, L_U, p_L=0,  \\
	- r+a 			& if \ \  X_U, L_C, p_X=1,  \\
	- r 			& otherwise. 		\end{array} \right. 
	\end{align}
	
	We expect that $i$ vote for $SP$ iff $ v(CB) - v(SP) +\varepsilon_{i} < 0$.\footnote{We assume voters prefer $CB$ in case of indifference, so we use strict inequality.} Thus, the probability that citizen $i$ votes for $SP$ is:
	\begin{eqnarray} \label{eq:probSP}
	Pr(SP)=Pr(\varepsilon_{i} < v(SP)-v(CB))= F(v(SP)-v(CB))
	\end{eqnarray} 
	
	We can now derive several predictions of the model to be tested in the experiment.
	
	In Proposition \ref{prop:gridlockeffects}, we analyze the impact of political gridlock on the probability of $SP$. We show that, depending on the prevailing circumstances, political gridlock may raise or reduce the probability of $SP$ and provide precise characterizations of these circumstances. In order to test the impact of political gridlock on the probability of $SP$, we use the benchmark of no political gridlock. This proposition provides our first four hypotheses: H1 to H4. 
	
	\begin{proposition} \label{prop:gridlockeffects}{The effects of political gridlock.}	
		
		\begin{enumerate}
			\item Political gridlock (weakly) \textbf{raises} the probability of $SP$ iff the gridlock occurs with
			\begin{itemize}
				\item H1: biased $X$ and $L$ ($X_R$ and $L_C$) and the reform is ex ante beneficial ($q>1/2$), or 
				\item H2: unbiased $X$ ($X_U$).
			\end{itemize}
			\item Political gridlock (weakly) \textbf{reduces} the probability of $SP$ iff the gridlock occurs with 
			\begin{itemize}
				\item H3: biased $X$ and $L$ ($X_R$ and $L_C$) and the reform is ex ante harmful ($q\leq 1/2$), or
				\item H4: unbiased $L$ ($L_U$).\footnote{By ``weakly raises'' and ``weakly reduces'' we mean ``does not reduce'' and ``does not raise'', respectively.} 
			\end{itemize}
		\end{enumerate}	
		
	\end{proposition}

The next propositions study the effect of $q$ and $r$ on the probability of $SP$.
	
	\begin{proposition}  \label{prop:qeffects}{H5: The effects of $q$.}	
		
		The probability that voters grant $SP$ is a non-decreasing function of the probability that the reform is beneficial for citizens iff the executive is reformist and the legislature is conservative. Otherwise, the probability of $SP$ does not depend on the probability $q$.	
	\end{proposition}

	\begin{proposition}  \label{prop:renteffects}{H6: The effects of rents.}	
		
		The probability that voters grant $SP$ is not increasing in the amount of rents.

	\end{proposition}

	The model is mute regarding the effect of the framing and of the previous exposure of subjects to political gridlock, but we were interested in exploring these conditions as well, given the political baggage the subjects may come with to the laboratory. As we explain in further detail below, in five of the seven treatments the framing is neutral, with $X$ and $L$ presented as two decision makers, and $CB$ and $SP$ as rule 1 and rule 2, respectively. Also in these treatments $r$ is presented as the cost of rule 2. We conjectured that subjects might be less prone to granting $SP$ if they were told that $r$ are rents extracted by corrupt politicians rather than costs of the policy rule, even when their monetary gains were exactly the same with the corruption and neutral framing. Similarly, a political framing, that includes corruption plus the whole political wording, might have an impact on subjects willingness to grant $SP$. So we have the following additional hypotheses.
	
	\medskip
	
	\noindent \textit{H7: The probability of $SP$ is lower in the corruption than in the neutral framing.}
	\medskip
		
	\noindent \textit{H8: The probability of $SP$ is lower in the political than in the neutral framing.}
	\medskip

	\bigskip

	As we explain in detail in the next section, in the first part of the experiment where subjects are trained, they have to choose between $CB$ and $SP$ in ten different situations. We introduce variability in the frequency of political gridlock in the priming phase. The model predicts no impact of the history of political gridlock on subjects decisions. Indeed, if subjects understood the game perfectly and were rational, the history of stalemates in the priming phase should be irrelevant. However, subjects might behave differently depending on whether they were primed with frequent or infrequent political gridlock. 
	\medskip
	
		\noindent \textit{H9: Subjects exposed to a history of unfrequent political gridlock will be more willing to grant $SP$.}
	\medskip

Notice that while H1 to H6 rest on the assumption that individuals maximize the utility function depicted in Equation (\ref{eq:utility}), H7 to H9 depart from it.

	\medskip

\subsection{Discussion} \label{discussion}
In this section, we make some remarks regarding the assumptions of the model, and based on these remarks we discuss some extensions.
 
\paragraph{Effectiveness of checks and balances.} We assume that checks and balances are highly effective, implying zero rents when this mechanism is in place.\footnote{For this to be the case, the institutional design of separation of powers must guarantee appropriate opposition of interests  \parencite{Persson1997, Persson2000}.}  This assumption  simplifies illustrating the mechanism we want to highlight and, more importantly, provides a lower bound on the weakening of checks and balances caused by political gridlock. Indeed, we show that voters are willing to grant special powers even when checks and balances are extremely effective in curbing rent extraction. This preference would be even more pronounced with less effective checks and balances.

\paragraph{Dynamic consequences of checks and balances.} We consider costs linked to special powers within a single period. However, it is important to recognize that weakening checks and balances could have long-term consequences. Though we do not model these effects, the costs of rent extraction can be seen as a reduced form of the discounted expected costs of special powers. Future analysis might explicitly incorporate the decisions of voters and politicians in a dynamic model.

\paragraph{Accountability and the protection of minorities.} Separation of powers and checks and balances serve two primary functions in modern representative democracies. Firstly, it controls deviations by elected officials, such as corruption or pursuing private agendas \parencite[see, among others,][]{Persson1997, Persson2000, Stephenson2010, Forteza2019}. This ``horizontal accountability'' \parencite{ODonnell1998} or ``internal control'' \parencite{Besley2018} complements the ``vertical accountability'' or ``external control'' provided by elections. Secondly, separation of powers protects minorities against the ``tyranny of the majority''. While this was emphasized by classical advocates of separation of powers (Locke, Montesquieu, Madison, Tocqueville) and is also emphasized in several highly influential modern papers \parencite{Aghion2004, Hayek1960, Buchanan1962, Buchanan1975a, Maskin2004}, our focus remains on scenarios where the protection of minorities is not a prominent concern \parencite[see][for a similar choice]{Acemoglu2013c, Forteza2019}. Hence, in the present paper we abstract from this theme.

\paragraph{Strategic setting.} In our model, we assume that both $X$ and $L$ are not strategic agents. A conservative (reformist) politician consistently proposes the status quo (reform), while an unbiased politician proposes policies aligned with the state of nature. This simplification allows us to concentrate on voters' decisions. Nonetheless, extending the model to incorporate strategic behavior by $X$ and $L$ is feasible. In \textcite{Forteza2019}, a conservative legislature supports the status quo, and an executive endogenously determines proposed policies. This strategic setup introduces endogenous policy proposals and strategic interaction between voters and politicians. In this setting, a reform is only possible when it is proposed by the executive, and voters grant special powers. The former occurs only when the executive is not strongly pro-status quo biased, and the latter when voters’ expected gains from reform are larger than the costs of rent extraction. 

In this extended setting, the choices of a politician $j$ in the executive comes from the maximization of the following payoffs:
$$
u_j=-(p-s-\delta_j)^2+ a_j r_X -b r_{L},
$$
where $p$ is the implemented policy, $s$ the state of nature, $\delta_j \in \mathbb{R}$ the bias of the politician, $a_j>0$, $r_X$ the rents extracted by the executive, $b>0$, and $r_L$ the rents extracted by the legislature. The key parameter in this setting is the bias of the politician in the executive, $\delta_j$. For a sufficiently negative (positive) $\delta_j$, the politician always proposes $p=0$ ($p=1$). For intermediate values, the policy proposal depends on the state of nature. The qualitative results of this model closely mirror those we presented earlier.

\paragraph{Commitment.} We assume that politicians can commit to the policies they announce. As extensively discussed in \textcite{Forteza2019}, it is possible to extend the model to study a situation where politicians cannot commit to their platforms. Although the results in this context depend on the bias of the executive, the qualitative results regarding the occurrence of special powers remain consistent with those we presented previously.

\paragraph{Elections.} In the current model, we do not delve into the process of government election. When considering the identities of the politicians in the executive and legislature, one might question why the electorate would choose a conservative legislature and subsequently grant the executive special powers for further reforms. Would not it be preferable to avert political gridlock during elections? By doing so, voters could prevent gridlock without dismantling controls over the executive. In \textcite{Forteza2019}, an equilibrium is demonstrated for intermediate values of $q$ where voters opt for a reformist executive and a conservative legislature. In this equilibrium, voters refrain from granting the president the legislative majority required to implement their agenda. And yet they later vote for special powers with the added cost of rent extraction to advance the same proposed agenda.  This behavior stems from the fact that, while voters only know that $s=1$ occurs with probability $q$ during the election period, they observe the realized state of nature at the time they have to decide between checks and balances and special powers.

	\section{Experimental design}\label{design}
	
	The experiment implements the theoretical framework described above where subjects take the role of voters. As neither $X$ or $L$ in our model are strategic agents, we design a decision making experiment without strategic interactions. Therefore, $X$ and $L$ in the experiment are not real subjects. There is no decision made by $X$ and $L$ as conservative, reformist, and unbiased types always propose $p=0$, $p=1$, and $p=s$, respectively, and these types are observed by subjects. This simplifies the decision environment, and allows us to focus on the impact of political gridlock on subjects' support to the loosening of checks and balances.

We have seven treatments, which vary based on the framing of the instructions and the combination of parameters in the model. As this is the first experimental design to address this question, we lack guidance regarding parameter values from previous studies. We simulated the theoretical model to choose a set of parameter values that should be detectable in our experiment, provided the variance of the preference parameter $\varepsilon$ in the  pool of subjects is not much higher than the variance assumed in the simulations.
Table \ref{table:treatments-between} summarizes the conditions that define each treatment (as well as the number of sessions conducted and the number of participants in each treatment).

	\begin{table}[H]
		\centering
		\begin{threeparttable}
			\caption{Treatments}
			\begin{tabular}{ c c c c c c c }
				\toprule
				\multicolumn{1}{p{5.39em}}{Treatment} & Rents & \multicolumn{1}{p{3cm}}{\centering Probability reform  \\ is beneficial} &  \multicolumn{1}{p{3cm}}{\centering Frequency of \\ gridlock} & Framing & \# Sessions & \#Subjects \\
				\midrule
				1     & $r_L<a$ & $q_H$ & $0.3$   & neutral & 2 & 26 \\
				\midrule
				2     & $r_L<a$ & $q_L$ & $0.2$   & neutral & 1 & 16 \\
				\midrule
				3     & $r_L<a$ & $q_H$ & $0.6$  & neutral & 4 & 40 \\
				\midrule
				4     & $r_L<a$ & $q_L$ & $0.6$  & neutral & 5 & 49\\
				\midrule
				5     & $r_H>a$ & $q_H$ & $0.6$  & neutral & 7 & 46\\
				\midrule
				6     & $r_L<a$ & $q_H$ & $0.6$  & corruption & 5 & 33 \\
				\midrule
				7     & $r_L<a$ & $q_H$ & $0.6$  & political & 6 & 33 \\
				\bottomrule
			\end{tabular}
			\label{table:treatments-between}
		\end{threeparttable}
				\Fignote{Notes: $r_H$ and $r_L$ denote high and low rents, respectively. $q_H$, and $q_L$ denote high and low probability of $s=1$, respectively. In particular, we assume for the experiment the following values: $r_H=96$, $r_L=24$, $q_H=0.9$, $q_L=0.2$, and $a=80$.}
	\end{table}
				
	In the first five treatments, the framing is neutral: the executive and the legislature are introduced as two decision makers ($X$ and $L$), $CB$ are referred as ``rule 1'', $SP$ as ``rule 2'', and rents are explained as the cost of adopting rule 2  \parencite[see][for similar neutral wording]{Agranov2015, Agranov2018,DalBo2010, Ryvkin2017}. The cost of using rule 2 is low in the first four treatments, and we vary the probability that the reform is beneficial and the frequency of gridlock in the priming phase. Treatment five has a high cost of implementing rule 2.
		
	Treatment 6 has the same parameter combination as Treatment 3 with the difference that the cost of rule 2 is framed as the loss subjects have from corruption (the rents stolen by $X$). Only in Treatment 7 we frame everything in terms of political decisions in a presidential system: decision makers are the executive and the legislature, the subjects are citizens voting on checks on the executive, rules 1 and 2 are $CB$ and $SP$, respectively, and the cost of special powers is corruption \parencite[see][for a similar wording]{Leight2018}. These two treatments are introduced to understand how framing impacts politically-related decisions. 
	
	Each treatment is divided into three different stages, described in detail below.

	Our experimental design allows us to make comparisons both between and within-subjects. The within-subject component of our design is aimed at reducing the variance of unobserved effects, increasing the precision of the estimation of treatment effects \parencite{List2011}. 
	
	\subsection{Stage 1: Training + priming}
		
	In stage 1 of the experiment, subjects go through a training stage, where the decision making is explained and they face a series of scenarios they have to understand and respond to. Subjects face 10 different scenarios, each one consisting of two parts.
	
	In the first part, subjects have to (i) predict policy proposals knowing politicians' types and the state of nature, (ii) tell whether there is political gridlock, and (iii) predict policies with and without $SP$. The answers to these questions are either correct or incorrect. This part is also both a learning and a priming phase, as different treatments expose subjects with different scenarios that have varying states of nature and types of politicians. For the purposes of both learning and incentivising subjects to correctly respond in each scenario, if subjects incorrectly answer the questions of a specific scenario (after four tries), the program automatically shows them the correct answers and allows them to move to the next part (and are penalized with a payment of 0 in that round). 
		
	In the second part of each scenario, subjects have to choose between $CB$ and $SP$ (or rule 1 or 2 depending on the treatment). 
	For the purposes of decision-making, the decision screen gives subjects the opportunity to consult the cost associated to each rule (remember that the cost associated to rule 1 is $0$ but for rule 2 the cost is positive).		
	
	At the end of this stage, subjects are asked about the frequency of political gridlock. They are presented with a question that ranges from 0 to 100 percent (in brackets of 10). 
	
	\subsection{Stage 2: Uncertainty about the state of nature} \label{section:stage2}
	
	The second stage is the relevant one to test the predictions of the model. Here, participants do not observe the state of nature. They only know the probability of $s=1$ (i.e. $q$) and have to choose a rule based on this information. Subjects make decisions in 14 scenarios, each consisting of two parts.  
	
	In the first part, subjects are informed of the policy proposals and types of politicians, but not of the state of nature. Subjects have to answer: (i) whether there is political gridlock or not, and (ii) the policy implemented with each rule.\footnote{Notice that while we call players attention on the possibility that the legislature blocks a reform, we never tell subjects they benefit from granting $SP$. We do however, remind them in each case that choosing $SP$ has a cost.} The answers to these questions are either correct or incorrect. If subjects incorrectly answer the questions of a specific scenario (after four tries), the program automatically shows them the correct answers and allows them to move to the next part (and are penalized with a payment of 0 for that scenario).  
	
	In the second part, subjects have to choose between $CB$ and $SP$ (or rule 1 or 2 depending on the treatment). For the purposes of decision-making, the decision screen gives subjects the opportunity to be reminded of the cost associated to each rule. 
	
	The 14 scenarios are all the possible consistent combinations of $X$ and $L$ types and proposals, and are common to all treatments. In Table \ref{table:SPgains}, we present the characteristics associated with each of the 14 scenarios, and summarize the experimenter's rational expectation of citizens expected gains from $SP$ in each treatment, i.e. $\mathbb{E}[v(SP)-v(CB)+\varepsilon_{i}]=v(SP)-v(CB)$.

	\begin{table}[H]
		\centering
		\caption{Expected net gains from $SP$ across treatments}  \label{table:SPgains}%
		\begin{threeparttable}
			\begin{tabular}{cccccccrrrrrrr}
				\cmidrule{1-14}  \multicolumn{1}{p{4em}}{\multirow{2}[4]{*}{Scenarios }} & \multicolumn{2}{c}{Politicians types} &       & \multicolumn{2}{c}{Proposals} &       & \multicolumn{7}{c}{Expected net gains from $SP$ in each treatment} \\
				\cmidrule{2-3}\cmidrule{5-6}\cmidrule{8-14}    \multicolumn{1}{c}{in Stage 2}    & X     & L     &       & \multicolumn{1}{c}{$p_X$} & \multicolumn{1}{c}{$p_L$} &       & \multicolumn{1}{c}{1} & \multicolumn{1}{c}{2} & \multicolumn{1}{c}{3} & \multicolumn{1}{c}{4} & \multicolumn{1}{c}{5} & \multicolumn{1}{c}{6} & \multicolumn{1}{c}{7} \\
				\midrule
				1     & Conservative     & Conservative     &       & 0     & 0     &       & -24   & -24   & -24   & -24   & -96   & -24   & -24 \\
				2     & Conservative     & Reformist     &       & 0     & 1     &       & -24   & -24   & -24   & -24   & -96   & -24   & -24 \\
				3     & Conservative     & Unbiased     &       & 0     & 0     &       & -24   & -24   & -24   & -24   & -96   & -24   & -24 \\
				4     & Conservative     & Unbiased     &       & 0     & 1     &       & -24   & -24   & -24   & -24   & -96   & -24   & -24 \\
				5     & Reformist     & Conservative     &       & 1     & 0     &       & 40    & -72   & 40    & -72   & -32   & 40    & 40 \\
				6     & Reformist     & Reformist     &       & 1     & 1     &       & -24   & -24   & -24   & -24   & -96   & -24   & -24 \\
				7     & Reformist     & Unbiased     &       & 1     & 0     &       & -104  & -104  & -104  & -104  & -176  & -104  & -104 \\
				8     & Reformist     & Unbiased     &       & 1     & 1     &       & -24   & -24   & -24   & -24   & -96   & -24   & -24 \\
				9     & Unbiased     & Conservative     &       & 0     & 0     &       & -24   & -24   & -24   & -24   & -96   & -24   & -24 \\
				10    & Unbiased     & Conservative     &       & 1     & 0     &       & 56    & 56    & 56    & 56    & -16   & 56    & 56 \\
				11    & Unbiased     & Reformist     &       & 0     & 1     &       & -24   & -24   & -24   & -24   & -96   & -24   & -24 \\
				12    & Unbiased     & Reformist     &       & 1     & 1     &       & -24   & -24   & -24   & -24   & -96   & -24   & -24 \\
				13    & Unbiased     & Unbiased    &       & 0     & 0     &       & -24   & -24   & -24   & -24   & -96   & -24   & -24 \\
				14    & Unbiased     & Unbiased     &       & 1     & 1     &       & -24   & -24   & -24   & -24   & -96   & -24   & -24 \\
				\bottomrule
			\end{tabular}%
		\end{threeparttable}
	\end{table}%
	
	Subjects might experience order effects ---learning or getting tired--- along the experiment. To avoid confounding the order and the treatment effects, experimenters sometimes randomize the order  \parencite{List2011}. In our design we can control for order effects comparing the gridlock and non-gridlock rows that come both before and after the gridlock row. The effect of order can then be identified without randomizing the order of treatments. We have no evidence of such effects. Even though we did not perform a formal systematic analysis of order effects, we did check that our results are not sensitive to the inclusion of a control for the order of tasks in the experiment. 

	\subsection{Stage 3: Post-experimental survey}
	
		After the experiment is finished, subjects respond to a questionnaire. We capture a measure of individual risk aversion (or tolerance to risks) as well as demographic and socio-economic questions (age, education, education of parents, gender), beliefs about different relevant topics, such as income distribution, competition, political leadership, and self-placement on an income scale. Most of the questions are taken from the world values survey (wvs), which facilitates the comparison of our sample with country wide survey results. To measure risk aversion, subjects are presented with a Multiple Price List (MPL) \parencite{Holt2002}.
		We use these variables as controls in our analysis, to account for the political and socio-economic baggage our subjects might come with into the lab. The questionnaire, including the risk aversion task can be found in the online Appendix (Section \ref{Ques}). 
	\subsection{Experimental procedure}
	
	We conducted all experimental sessions in 2019, between August and October, at the Experimental Laboratory of the Faculty of Social Sciences at Universidad de la Rep\'ublica, Uruguay.  Participants were randomly assigned to one of seven treatments. Each session consisted of the same treatment. Instructions included a set of tables that explained how decision makers actions, proposals and the state of nature could be combined, as well as the payoffs for each combination. An example of those instructions can be found in the online Appendix (Section \ref{Instr}). Subjects were recruited using the online recruitment program ORSEE \parencite{Greiner2015}. We implemented the experiment using z-Tree \parencite{Fischbacher2007}. Questions before and during the experimental sessions were answered in private at the subject's workspace by the experimenters.
	  
	We conducted 30 sessions, with a total of 243 subjects. No subject participated in more than one session. The experiment lasted, on average, 90 minutes, and subjects' minimum and maximum earnings were equivalent to $3.5$ and  $15.4$ American dollars, respectively.\footnote{These earnings correspond to 130 and 568 Uruguayan pesos, respectively. The minimum wage in Uruguay in 2019 was 2 American dollars per hour.} 
	The final payment consists of one randomly selected decision for stage 1, one for stage 2 and one for Stage 3, plus a fee for filling out the post-experimental questionnaire and a reimbursement for transportation costs (the last two stages of the payment are in lieu of a show-up fee). The payments for stages 1 and 2 depend on the individual choices between rules 1 and 2, while the payment for Stage 3 depends on the option chosen by the individual in the selected row (A or B) and a random draw based on the probabilities of each potential outcome in the chosen option. The decisions selected for payment in Stages 1, 2 and 3 and the outcome draws for Stage 3 were randomly done by the computer during the experiment. As cash payments in experiments are not allowed in Uruguay, subjects were paid at the end of the experimental session, in private, with a gift card that could be used in one of the supermarkets in Uruguay with the largest number of branches.

    \section{Methods and Results}\label{results}          

In this section, after a brief description of the methods, we present and discuss our results. Additionally, descriptive statistics of the results of the experiment, and an analysis of balance of some covariates between treatment arms can be found in sections \ref{descriptive} and \ref{balance} of the Appendix, respectively.
	
	\subsection{Hypothesis testing}
  
  We observe subjects binary choices over rules.  Each subject is asked to choose between the two rules fourteen times, so the data has a panel structure with 14 periods. The characteristics of politicians and policy proposals varied across the fourteen periods, generating different environments or hypothetical situations, as indicated in Table \ref{table:SPgains}. We exploit this longitudinal variation to study the impact of political gridlock on individuals choices in different environments. 
   
  To test our main hypotheses, we regress the binary variable $SP$ on dummies that adopt value one when the conditions defining each hypothesis are fulfilled. Specifically, we run simple linear regressions of the following form:
\begin{equation} \label{eq:sp1}
SP_{it} = \sum_{k=1}^{k=9} \beta_k Hk_{it} + \beta_x x_{it} 
\end{equation}
where: 
\begin{itemize}
\item $SP_{it}=1$ if subject $i$ chose rule 2 (SP) in period $t$, and zero otherwise; 	
\item $Hk_{it}=1$ if conditions that define hypothesis $k$ are fulfilled in observation $it$, zero otherwise; and
\item $x_{it}$ is a set of controls described in Table \ref{table:main}. 
\end{itemize}

  Table \ref{table:treatmentsandcontrols} summarizes the conditions of the treatment and control groups in each of our nine main hypotheses. 
  
We also explore differences in the response to political gridlock with harmful reforms among subgroups of participants. To do that, we run the regression (\ref{eq:sp2}), where we interact the $H3_{it}$ and $H4_{it}$ variables with a set of dummy variables $z$ that divide subjects along several dimensions, namely (i) whether the subject made mistakes in part 2, (ii) risk aversion, (iii) gender, (iv) ideological self identification, (v) preference for strong leader, and (vi) framing (corruption and political framing).\footnote{See the note in Table \ref{table:main} for a detailed description of these variables}
  
  \begin{equation} \label{eq:sp2}
  	SP_{it} = \sum_{k=3}^{k=4} \beta_{kz} Hk_{it}z_{it} + \sum_{k=3}^{k=4} \beta_k Hk_{it} + \beta_z z_{it} + \beta_x x_{it} 
  \end{equation}
  
We run OLS regressions, and control for multiple hypotheses testing (mht) looking at Family-wise error rate (FWER) \parencite{Barsbai2020, List2019} and False discovery rate (FDR) \parencite{Anderson2008}.\footnote{FWER routines control the probability that there is one or more false rejections in the set of hypotheses. FDR routines focus on the frequency of false rejections. The former tends to be more conservative in the sense that it is usually less likely to reject H0 with FWER than FDR routines, but this implies an increased probability of type II errors.  \textcite{McKenzie2020} presents a very helpful overview of multiple hypothesis testing commands in STATA.} The mht issue arises in our setting primarily due to the several arms of the experiment. It also arises in an exploratory analysis we conducted to study an unexpected result (see section \ref{excesssp}). We do not have mht stemming from several outcomes or several estimators. In any case, as it is common practice, we also report the unadjusted p-values.

  Random assignment of subjects to treatments minimizes the risk of confounding the impact of the treatment-defining characteristics with individual idiosyncratic preferences for rules (and other individual traits), and thus allows us to identify the impact of ``time''-invariant treatments on the variable of interest.\footnote{Running panel fixed effect regressions with our data would impede the identification of these effects without adding consistency.} We cluster at subjects level to account for the fact that each subject plays several rounds.\footnote{We briefly comment below on the effects of clustering at session rather than individual level.}     
   
 \begin{table}[htbp]
 	\centering
 	\caption{Treatments and controls}
 	\begin{tabular}{lll}
 		\toprule
 		Hypothesis & \multicolumn{1}{l}{Treatment} & Control \\
 		\midrule
 		H1    &  $p_X=1, p_L=0, X_R, L_C, q>0.5$   	& $p_X=0 \ \& \ p_L=0$	or $p_X=0 \ \& \ p_L=1$ or $p_X=1 \ \& \ p_L=1$ \\
 		H2    &  $p_X=1, p_L=0, X_U$ 				& $p_X=0 \ \& \ p_L=0$	or $p_X=0 \ \& \ p_L=1$ or $p_X=1 \ \& \ p_L=1$	\\
 		H3    &  $p_X=1, p_L=0, X_R, L_C, q<0.5$  	& $p_X=0 \ \& \ p_L=0$	or $p_X=0 \ \& \ p_L=1$ or $p_X=1 \ \& \ p_L=1$	\\
 		H4    &  $p_X=1, p_L=0, L_U$  				& $p_X=0 \ \& \ p_L=0$	or $p_X=0 \ \& \ p_L=1$ or $p_X=1 \ \& \ p_L=1$	\\
 		H5    &  $p_X=1, p_L=0, X_R, L_C, q>0.5$	& $p_X=1, p_L=0, X_R, L_C, q<0.5$	\\
 		H6    &  $r_H$								& $r_L$	\\
 		H7    &  $T=6$ (corruption framing)			& $T \neq 6$	\\
 		H8    &  $T=7$ (political framing)			& $T \neq 7$	\\
 		H9    & $T \leq 2$ & $T >2 $	\\
 		\bottomrule
 	\end{tabular}%
 	\label{table:treatmentsandcontrols}%
 \end{table}%
 
As a robustness check, we also compute Fisher tests of differences in frequencies of $SP$ between specific treatments that differed in one and only one dimension. This approach provides very clean comparisons of frequencies between treatments keeping all else equal, but at the cost of a dramatic drop in the number of observations in each test. While the Fisher test is appropriate with this type of data, the power is obviously smaller with the Fisher tests than with the regressions. In the next section, we present and discuss the results obtained with the regression analysis. The results obtained using Fisher tests are presented in the online Appendix (Section \ref{Fisher}). The main results are the same with both approaches.

	\subsection{Main results}
	
We summarize our main results in Table \ref{table:main}. The second column contains the difference in the frequency of SP under the treatment indicated in the first column and in the control group. These are our parameters of interest. The third to fifth columns present the p-values of these differences without and with mht adjustments. The last column contains the frequency of SP in the control groups. So, for example, introducing treatment H1 rose the frequency of SP by 31 percentage points from a baseline of 10.2 per cent.   

The first four rows (H1 to H4) report the difference in frequencies of $SP$ that are associated to political gridlock (the $\beta_k$ coefficients in (\ref{eq:sp1})). The control group is the set of decisions in which subjects were faced with no political gridlock.

	\begin{table}[htbp] 		
		\centering
		\begin{threeparttable}
		\caption{Hypothesis testing: the impact of treatments on the frequency of $SP$.}\label{table:main}
    \begin{tabular}{lrrrrr}
	\toprule
	& \multicolumn{1}{l}{Difference} & \multicolumn{3}{c}{p-values} & \multicolumn{1}{l}{Control } \\
	\cmidrule{3-5}          &       & \multicolumn{1}{l}{Unadjusted} & \multicolumn{2}{c}{mht adjusted} & \multicolumn{1}{l}{mean} \\
	\cmidrule{4-5}          &       &       & \multicolumn{1}{l}{Barsbai et al.} & Anderson &  \\ \hline
        H1 & 0.310 & 0.000 & 0.000 & 0.003 & 0.102 \\ 
        H2 & 0.490 & 0.000 & 0.000 & 0.003 & 0.102 \\ 
        H3 & 0.187 & 0.003 & 0.041 & 0.013 & 0.102 \\ 
        H4 & 0.097 & 0.000 & 0.000 & 0.003 & 0.102 \\ 
        H5 & 0.098 & 0.187 & 0.945 & 0.531 & 0.347 \\ 
        H6 & -0.072 & 0.053 & 0.570 & 0.165 & 0.179 \\ 
        H7 & -0.029 & 0.423 & 0.995 & 0.944 & 0.168 \\ 
        H8 & -0.063 & 0.018 & 0.250 & 0.064 & 0.172 \\ 
        H9 & 0.029 & 0.494 & 0.994 & 0.944 & 0.161 \\  
	\bottomrule
\end{tabular}%

			\Fignote{Notes: Number of observations = 3038. Added controls are: (i) mistakes (dummy variable indicating whether the individual needed to make more than one attempt to answer the correct-incorrect questions), (ii) risk averse (dummy variable indicating whether the individual is risk averse measured using a multiple price list, see Holt and Laury, 2002), (iii) female dummy, (iv) right wing dummy ($=1$ if the individual ideological self identification lies to the right of the median in the experiment), (v) strong leader dummy ($=1$ if the subject chose good or very good to the question regarding the convenience of having a strong leader), (vi) corruption dummy ($=1$ in the treatment in which costs of rule 2 were presented as corruption), and (vii) political framing dummy ($=1$ in the treatment in which the exercise was presented in political terms). Some of these are control variables in the testing of one or more hypotheses, and the treatment variable in the testing of another hypothesis. This is the case, for example, of the political framing dummy, which is the treatment variable in H8. As a robustness check, we ran the same regression substituting the ten- for the two-point scale estimates for risk aversion and ideological self-identification and obtained the same results. 
	
  }
		\end{threeparttable}
	\end{table}%

	As expected, in our experiment political gridlock caused by biased politicians raised the probability of $SP$ when the reform was ex ante beneficial (H1). This type of gridlock caused a 31 percentage point rise in the frequency of political gridlock relative to the no gridlock case. Also, political gridlock caused by unbiased executives raised the frequency of $SP$ by 49 percentage points (H2). These results are statistically highly significant at the usual significance levels.   
	
	Unexpectedly, subjects also granted $SP$ in higher proportions when faced with gridlock caused by biased politicians even when the reform was ex ante harmful. This type of gridlock caused an almost 19 percentage point rise in the frequency of $SP$. This result contradicts our hypothesis 3.\footnote{Clustering at session rather than individual level, and using mht adjusted p-values, this coefficient became not significant at conventional levels. Other qualitative results did not change.} Also, subjects granted $SP$ in higher proportions when faced with gridlock even when the legislature was unbiased and truthfully ``warned'' that the reform was not beneficial. This gridlock caused an almost 10 percentage point rise in the frequency of $SP$, contradicting our hypothesis 4. These effects are statistically significant at the usual significance levels (using the Barsbai et al mht-adjusted or the Anderson mht-adjusted pvalues, H3 would not be statistically significant at 1 percent, but it would still be significant at 5 percent).
	
	The picture that emerges from these results is thus only partially consistent with the predictions of our model: \textbf{subjects chose \textit{SP} more frequently with than without political gridlock, but they did it even when the reform was ex ante harmful.}  
	It is worth mentioning that these unexpected results seem not to be driven by subjects' misunderstanding of the problem. Indeed, individuals learned in the priming phase, as the number of attempts needed to respond correctly the right-wrong type of questions gradually reduced along this phase.
 Almost 88 per cent 
	of answers were correct on the first decision of the second stage of the experiment, suggesting that subjects understood the basics of the experiment quite well.  
These findings suggest that excess of special powers cannot be explained by lack of understanding of the game. 
	
	In the fifth row of Table \ref{table:main} we present a test of H5. According to this hypothesis, voters willingness to grant $SP$ in the presence of political gridlock caused by biased politicians is increasing (more precisely, not decreasing) in the probability that the reform is beneficial. We set this probability at 0.2 and 0.9 in the control and treatment groups, respectively. The observed frequency of $SP$ in the control group was 34.7 percent. The treatment caused a 9.8 percentage point rise in the estimated frequency of $SP$. While this represents a 28 percent increase relative to the control, it is not statistically significant.  
	 
	Hypothesis H6 says that subjects willingness to grant $SP$ is not increasing in the amount of rents the executive can extract. The difference between rents in the control and treatment groups was set at about 13 percent of the maximum payoff subjects could obtain in the whole experiment. 
	High rents induced lower frequency of $SP$ in more than 7 percentage points, on average. The estimated effect looks ``large'', compared to the control mean of  17.9 percent, but the estimation is imprecise and we cannot reject the null of no effects.  
	 
    Subjects were on average less willing to grant $SP$ in the corruption (H7) and political (H8) framing treatments than in the control group, but the differences are not statistically significant at 5 percent according to the Barsbai et al mht- and the Anderson mht-adjusted pvalues.

	Finally, exposing subjects to a low frequency of political gridlock in the priming phase caused an almost 3 percentage point increase in the frequency of $SP$, but this effect is not statistically significant at the usual significance levels.
	 
	 \subsection{Excess of special powers} \label{excesssp}
	 
	 Our experiment was designed to study how actual subjects deal with the trade-off between delegation and control. We showed that, as expected, subjects were willing to loosen control to facilitate reform, but they overreacted: they weakened checks and balances even when the reform in the government agenda was not beneficial, given that there is a cost of granting $SP$. In this sense, voters were excessively willing to grant special powers or, in short, there is an excess of special powers in our results. In this section we revise some possible factors that might be driving this unexpected result. 
	 
	 In Table \ref{table:excessSP}, we present estimations of the differences among subgroups of subjects in the response to the two type of political gridlock that caused excess of $SP$, based on equation \ref{eq:sp2}.  
	 Notice these are differences in differences, and the control mean is the average response to political gridlock ---i.e. the difference in the frequency of $SP$ with and without gridlock--- in the control group.\footnote{The results in Table \ref{table:excessSP} in the column entitled `Control mean' are estimations of $\beta_k =\mathbb{E}[SP_i|Hj=1,z=0] - \mathbb{E}[SP_i|Hj=0,z=0] $ and in the column entitled `Difference' are estimations of $\beta_{kz}=[\mathbb{E}[SP_i|Hj=1,z=1] - \mathbb{E}[SP_i|Hj=0,z=1]] - [\mathbb{E}[SP_i|Hj=1,z=0] - \mathbb{E}[SP_i|Hj=0,z=0]],  j\in\{3,4\} $.  }
	 
	 \begin{table}[htbp]
	 	\centering
		\begin{threeparttable}
		\caption{Excess of special powers.}\label{table:excessSP}%
	     \begin{tabular}{lrrrrr}
	 	\toprule
	 	& \multicolumn{1}{l}{Difference} & \multicolumn{3}{l}{p values} & \multicolumn{1}{l}{Control } \\
	 	\cmidrule{3-5}          &       & \multicolumn{1}{l}{Unadjusted} & \multicolumn{2}{l}{Multiple hyp. testing adjusted} & \multicolumn{1}{l}{mean} \\
	 	\cmidrule{4-5}          &       &       & \multicolumn{1}{l}{Barsbai et al} & Anderson &  \\
	 	\midrule
        H3*mistakes & 0.017 & 0.907 & 1.000 & 1.000 & 0.212 \\ 
        H4*mistakes & -0.065 & 0.443 & 0.994 & 0.944 & 0.107 \\ 
        H3*risk averse & 0.016 & 0.922 & 0.995 & 1.000 & 0.204 \\ 
        H4*risk averse & -0.017 & 0.834 & 1.000 & 1.000 & 0.112 \\ 
        H3*female & 0.136 & 0.236 & 0.957 & 0.547 & 0.147 \\ 
        H4*female & -0.001 & 0.987 & 0.987 & 1.000 & 0.096 \\ 
        H3*right wing & -0.102 & 0.374 & 0.994 & 0.944 & 0.264 \\ 
        H4*right wing & -0.022 & 0.686 & 1.000 & 1.000 & 0.108 \\ 
        H3*strong leader & -0.183 & 0.208 & 0.944 & 0.531 & 0.250 \\ 
        H4*strong leader & 0.012 & 0.858 & 1.000 & 1.000 & 0.094 \\ 
        H4*corruption & 0.100 & 0.188 & 0.933 & 0.531 & 0.081 \\ 
        H4*political framing & -0.013 & 0.876 & 1.000 & 1.000 & 0.093 \\ 

	 	\bottomrule
	 \end{tabular}%
	 
 			\Fignote{See footnotes in table \ref{table:main}.}
		\end{threeparttable}
	 \end{table}%

	 \paragraph{Mistakes.} We first explore whether the excess of $SP$ was driven by individuals who made more mistakes, i.e. individuals who had a poorer understanding of the tasks. To explore this possibility, we take advantage of a recorded registry of mistakes committed by subjects in the experiment.  
	 
	 We do not find statistically significant differences in the impact of political gridlock with harmful reforms on the frequency of $SP$ chosen by subjects who did and did not make mistakes (rows 1 and 2 in Table \ref{table:excessSP}). When confronted with political gridlock caused by biased politicians and harmful reforms, subjects who did and did not make mistakes raised the frequency of $SP$ by about 0.231 and 0.212, respectively, so the difference in response is in the order of 0.017, which is not statistically significant (see the first row of Table \ref{table:excessSP}). Similarly, the difference in the response to political gridlock with an unbiased legislature was not significantly different from zero.	 
	
	 \paragraph{Risk aversion.} More than 82 per cent 
	 of subjects in our experiment were risk averse. They were on average less willing to grant $SP$ (0.13) than other subjects (0.31).
 	 However, we do not find statistically significant differences in the frequency of $SP$ between risk averse and other subjects in the presence of political gridlock with harmful reforms or with an unbiased legislature (rows 3 and 4 in Table \ref{table:excessSP}, respectively).  
	 
	\paragraph{Female.} Women were on average more willing to grant SP than men (first row in Table 9 in the online Appendix), but they were not statistically significantly more responsive to political gridlock with harmful reforms than men.
	
	\paragraph{Right wing.}Subjects to the right of the median self-identification in the left-right ideological line were on average more willing to grant $SP$ than subjects to the left (row 8 in Table \ref{table:descriptive2} in the online Appendix). The reverse was true in the environment of political gridlock with harmful reforms, but the differences are not statistically significant (rows 7 and 8 in Table \ref{table:excessSP}). 
	
	\paragraph{Strong leader.} As expected, subjects who said that having a strong leader who disregards the congress and elections was good or very good, granted $SP$ more frequently than subjects who did not support that claim. However, we cannot reject that their response to political gridlock with harmful reforms is the same as of other subjects. 
	
	\paragraph{Corruption and political framing.} Framing the cost of $SP$ as corruption and the whole environment in political terms did not cause statistically significant changes in the response to gridlock with unbiased legislatures.\footnote{The experiment did not include treatments with corruption and political framing and a low probability that the reform is beneficial. Hence, we have no data to compute the impact of corruption and political framing on the response to gridlock with biased politicians and harmful reforms (H3). }
	
	A limitation we faced in the analysis of factors driving excess of $SP$ was lack of power. For example, we only have 12 observations with the conditions of H3 and more than one check of the correct-incorrect answers, and this is far from the minimum number of observations needed to test a difference in proportions lower than 2 percent. In this environment, the frequency of $SP$ was 0.212 among subjects who provided the correct answers at the first attempt, and 0.229 among subjects who needed more attempts. More than seven thousand subjects per group would be necessary to significantly detect this difference in proportions. This lack of statistical power is a direct consequence of the unexpected nature of the excess of special powers. Being this result totally unexpected, we did not design the experiment to test specific hypotheses regarding factors driving this result. Our analysis in this section is thus exploratory.	
	   
	\section{Concluding remarks}\label{concluding}
	
	In this paper we present the results of a lab experiment in which subjects were asked to choose between two rules that resemble checks and balances and executive special powers. Under checks and balances, the legislature can block a reform proposed by the executive. Under special powers, the will of the executive prevails, so there is no political gridlock.
	
	As expected, political gridlock emerged as an important driver of special powers. Subjects in the experiment were very willing to grant special powers in the presence of political gridlock, and they did it not only when the reform was beneficial but also when it was harmful. In this sense, there was an excess of special powers caused by political gridlock in our experiment.  
	
	The excess of special powers arose in two cases. First, when both politicians are biased ---so no information can be elicited from their proposals--- and the reform is ex ante harmful. Second, when an unbiased legislature proposes the status quo policy and a biased executive proposes reform.   

	In the first case, the probability that the reform matched the state of nature was only 20 percent and politicians proposals were not informative, so there should have been little doubt that granting special powers would most likely bring bad outcomes. And yet, many subjects voted for it. In the second case, subjects may have failed to realize that the legislature was revealing the true state of nature. However, in the symmetric case in which the executive was unbiased, subjects seem to have responded voting for special powers in higher proportions. So it does not seem to be the case that subjects totally failed to realize that unbiased politicians' proposals conveyed valuable information.

The results in this paper show how fragile checks and balances can be. It is enough in our experiment to ask subjects to report whether a reform proposed by the executive is being blocked by the legislature to induce them to vote for the weakening of checks and balances even when the reform is harmful. While we did not expect this result, we think it resembles many real world cases in which strong leaders convinced citizens to grant them special powers arguing the opposition was blocking much needed reforms. Because we did not anticipate the excess of special powers though, our design is not specially suited to study this deviation from optimal behaviour. Our experiment does provide a clear answer to the question of why voters weaken checks and balances: to remove a political gridlock. But not always removing a gridlock is beneficial and we  cannot say when or why the executive succeeds in making this case. We did not have enough statistical power to test heterogeneity in excess of $SP$ among subgroups of subjects. Nevertheless, the analysis in this paper provides some guidance for future lines of research on this important topic. In particular, once we know that participants may grant $SP$ in excess, we can design an experiment with two branches comparing the frequency of $SP$ with and without an explicit  argument that removing $CB$ facilitates reform. 

Finally, in the model and the experiment, we assume that politicians can commit to the policies they announce. This assumption allows us to focus on voters' behaviors. The interaction between commitment and the prevalence of special powers is an interesting topic for future research. Although \textcite{Forteza2019} present some theoretical advances in this direction, it would be interesting to incorporate a situation without commitment in the experimental design. 
		 
\printbibliography[
heading=bibintoc,
title={Bibliography}]

\newpage
\begin{center}
	\Large{\textbf{Online Appendix (not for publication)}}    
\end{center}
\setcounter{page}{1}
\setcounter{section}{0}

\bigskip

\section{Proofs of the propositions}

	\begin{proof}[Proof of Proposition \ref{prop:gridlockeffects}.]
	Political gridlock may arise under three different configurations of politicians types, corresponding to the three first lines in equation (\ref{eq:vSP-vCB}). There is no gridlock under any other circumstances. 
	Let $Pr(SP|gr)$ and $Pr(SP|No-gr)$ represent the probabilities of $SP$ with and without political gridlock, respectively. Then, using equations (\ref{eq:vSP-vCB}) and (\ref{eq:probSP}) we have that:\footnote{As we mentioned before, when there is no political gridlock, $v(CB) - v(SP) =r$. Then voters support SP iff  $\varepsilon_{it} < -r$.}
	\begin{align} \label{eq:gridlockonSP}
		Pr(SP|gr) - Pr(SP|No-gr)  =  \left\{  \begin{array}{ll} 	F(- r-a(1-2q)) - F(- r) 	& if \ \ X_R, L_C,  \\
			F(- r+a) - F(- r)  \geq 0	& if \ \  X_U, L_C, p_X=1,  \\
			F(- r-a) - F(- r) \leq 0		& if \ \  X_R, L_U, p_L=0.  \end{array} \right. 
	\end{align}	
	If the gridlock occurs with $X_R$ and  $L_C$, then the probability of $SP$ is greater with than without political gridlock iff $q>1/2$ (first line in equation (\ref{eq:gridlockonSP})). If the gridlock occurs with  $X_U$, $L_C$ and $p_X=1$, then the probability of $SP$ is greater with than without political gridlock (second line in Equation (\ref{eq:gridlockonSP})). If the gridlock occurs with $X_R$, $L_U$ and $p_L=0$, then the probability of $SP$ is lower with than without political gridlock (third line in equation (\ref{eq:gridlockonSP})).  
\end{proof}

	\begin{proof}[Proof of proposition \ref{prop:qeffects}.]
	Equation (\ref{eq:vSP-vCB}) implies that $v(SP)-v(CB)$ is an increasing function of $q$, if $X_R, L_C$, and does not depend on $q$ otherwise. Equation (\ref{eq:probSP}) says that $Pr(SP)$ is a non decreasing function of $v(SP)-v(CB)$.
\end{proof}

	\begin{proof}[Proof of proposition \ref{prop:renteffects}.]
	The hypothesis follows directly from equations (\ref{eq:vSP-vCB}) and (\ref{eq:probSP}).  
\end{proof}

\section{Descriptive statistics} \label{descriptive}

\paragraph{Who are the subjects?} We recruited 243 students from different faculties of the Universidad de la República, Uruguay. In Table \ref{table:descriptive1}, we summarize some characteristics and beliefs of this group of subjects and the country population, using our questionnaire and the results of the 2011 world values survey for Uruguay .

Our population is more feminized than the wvs sample (63 and 53 percent female, respectively). Sixty percent grew up in Montevideo,\footnote{More precisely, this is the percentage of the subjects responding they lived in Montevideo at the age of 10} the capital city, where about 39 percent of the country population lives. About 36 and 43 percent of their fathers and mothers, respectively, attended the university. Almost two thirds went to public primary and secondary schools. The percentage of subjects who think their income is above the country median is 46 in our experiment and 26 in the wvs sample.

\paragraph{What do they believe politically?} Not surprisingly, beliefs are not alike in our convenience sample of university students and the Uruguayan population.\footnote{Recall we posed the same questions as the wvs. There are however some differences worth mentioning. First, the questionnaire is administered by an interviewer in the wvs and self-administered ---i.e. the questions are presented in the screen of the computer--- in our experiment. Second, the interviewer in the wvs receives the following instruction: ``NOTE: Code but do not read out-- here and throughout the interview: 1. Don't know; 2. No answer; and 3. Not applicable''. Subjects in our experiment were allowed to skip a question, but we cannot distinguish the three options considered in the wvs protocol. We have no direct evidence on the impact that the different way of administering the survey may have had on responses. The percentage of answers in the lab was extremely high. Leaving aside 20 questionnaires entirely lost due to a software manipulation failure (with 223 questionnaires processed correctly), no question received less than 98.6 percent of answers.} 
Based on self ideological identification, our subjects tend to be more left-wing than the average population. Accordingly, they show less support than the wvs respondents to assertions such that ``we need larger income differences as incentives for individual effort'', ``people should take more responsibility to provide for themselves'' and it is good or very good to have a ``strong leader who does not have to bother with Congress and elections''. Our subjects also provided stronger support than wvs respondents to the assertion that ``government ownership of business and industry should be increased''. However, they show weaker support to the idea that ``competition is harmful'' and stronger support to the assertion that it is good ``having experts, not government, make decisions according to what they think is best for the country''.

\paragraph{How did they decide in our experiment?} In Table \ref{table:descriptive2} we present the proportion of cases in which several subgroups of subjects chose $SP$ rather than checks and balances. We obtained stronger support for $SP$ among females than males, and subjects whose parents did not have tertiary education. The frequency of $SP$ was higher among right- than left-wing self-identified participants, supporters of strong leaders and military government. It was lower among supporters of democratic government and individuals more interested in politics. 

\begin{table}[H]\centering \scriptsize 
	\begin{threeparttable}
		\caption{Descriptive statistics 1: some characteristics and beliefs of participants.}
		\vspace{-0.3cm}

     \begin{tabular}{lrcccccp{25.665em}}
 	\toprule
 	&       & \multicolumn{2}{c}{\textbf{Lab}} &       & \multicolumn{2}{c}{\textbf{WVS Uruguay}} & \multicolumn{1}{l}{\multirow{2}[4]{*}{\textbf{Definitions}}} \\
 	\cmidrule{3-4}\cmidrule{6-7}          &       & \textbf{mean} & \textbf{sd} &       & \textbf{mean} & \textbf{sd} & \multicolumn{1}{l}{} \\
 	\cmidrule{3-8}    \multicolumn{2}{l}{\textbf{CHARACTERISTICS}} &       &       &       &       &       & \multicolumn{1}{c}{} \\
 	\multicolumn{2}{l}{\textbf{Female}} & 0.63  & 0.48  &       & 0.53  & 0.50  & \multicolumn{1}{l}{= 0, if male; = 1 if female.} \\
 	\multicolumn{2}{l}{\textbf{Montevideo}} & 0.60  & 0.49  &       &       &       & \multicolumn{1}{l}{= 1, if raised in Montevideo (capital city); = 0 otherwise.} \\
 	\multicolumn{2}{l}{\textbf{Parents with tertiary education}} &       &       &       &       &       & \multicolumn{1}{r}{} \\
 	& \multicolumn{1}{l}{     Father } & 0.36  & 0.48  &       &       &       & \multicolumn{1}{l}{= 0, if no; = 1 if yes, even if incomplete.} \\
 	& \multicolumn{1}{l}{     Mother} & 0.43  & 0.50  &       &       &       & \multicolumn{1}{l}{= 0, if no; = 1 if yes, even if incomplete.} \\
 	\multicolumn{2}{l}{\textbf{Public Education}} &       &       &       &       &       & \multicolumn{1}{r}{} \\
 	& \multicolumn{1}{l}{     Primary } & 0.64  & 0.48  &       &       &       & \multicolumn{1}{l}{= 0, if no; = 1 if yes.} \\
 	& \multicolumn{1}{l}{     High school} & 0.64  & 0.48  &       &       &       & \multicolumn{1}{l}{= 0, if no; = 1 if yes.} \\
 	\multicolumn{2}{l}{\textbf{Income percentile (self perception)}} &       &       &       &       &       & \multicolumn{1}{r}{} \\
 	& \multicolumn{1}{l}{     Ten-point scale} & 5.32  & 1.81  &       & 4.50  & 1.81  & \multicolumn{1}{l}{= 1, if poorest; . . . ; = 10, if richest decile.} \\
 	& \multicolumn{1}{l}{     Two-point scale} & 0.46  & 0.50  &       & 0.26  & 0.44  & \multicolumn{1}{l}{= 0, if deciles 1 to 5; = 1, otherwise. } \\
 	\multicolumn{2}{l}{\textbf{BELIEFS}} &       &       &       &       &       & \multicolumn{1}{r}{} \\
 	\multicolumn{2}{l}{\textbf{Left- to right-wing}} &       &       &       &       &       & \multicolumn{1}{r}{} \\
 	& \multicolumn{1}{l}{     Ten-point scale} & 3.68  & 2.31  &       & 4.68  & 2.49  & \multicolumn{1}{l}{ = 0, if left; . . . ; =10 if right.} \\
 	& \multicolumn{1}{l}{     Two-point scale} & 0.19  & 0.39  &       & 0.28  & 0.45  & \multicolumn{1}{l}{= 0, if points 1 to 5 in the 10-point scale; = 1, otherwise. } \\
 	\multicolumn{2}{l}{\textbf{We need larger income differences}} &       &       &       &       &       & \multicolumn{1}{r}{} \\
 	& \multicolumn{1}{l}{     Ten-point scale} & 3.67  & 2.52  &       & 5.11  & 2.92  & = 1, if more equal is better; . . . ; = 10, if larger differences are needed. \\
 	& \multicolumn{1}{l}{     Two-point scale} & 0.22  & 0.41  &       & 0.40  & 0.49  & = 0, if points 1 to 5 in the 10-point scale; = 1, otherwise.  \\
 	\multicolumn{2}{l}{\textbf{Raise government ownership}} &       &       &       &       &       & \multicolumn{1}{r}{} \\
 	& \multicolumn{1}{l}{     Ten-point scale} & 6.43  & 2.60  &       & 5.56  & 2.62  & = 1, if private. . . ; . . . ; = 10, if government ownership should be increased. \\
 	& \multicolumn{1}{l}{     Two-point scale} & 0.65  & 0.48  &       & 0.43  & 0.50  & = 0, if points 1 to 5 in the 10-point scale; = 1, otherwise.  \\
 	\multicolumn{2}{l}{\textbf{People take more responsibility}} &       &       &       &       &       & \multicolumn{1}{r}{} \\
 	& \multicolumn{1}{l}{     Ten-point scale} & 3.79  & 2.60  &       & 5.07  & 3.04  & = 1, if government. . . ; . . . ; = 10, people should take more responsibility. \\
 	& \multicolumn{1}{l}{     Two-point scale} & 0.23  & 0.42  &       & 0.38  & 0.49  & = 0, if points 1 to 5 in the 10-point scale; = 1, otherwise.  \\
 	\multicolumn{2}{l}{\textbf{Competition is harmful}} &       &       &       &       &       & \multicolumn{1}{r}{} \\
 	& \multicolumn{1}{l}{     Ten-point scale} & 4.63  & 2.64  &       & 5.04  & 2.90  & = 1, if competition is good; . . . ; = 10, if competition is harmful. \\
 	& \multicolumn{1}{l}{     Two-point scale} & 0.31  & 0.46  &       & 0.37  & 0.48  & = 0, if points 1 to 5 in the 10-point scale; = 1, otherwise.  \\
 	\multicolumn{2}{l}{\textbf{Luck and contacts}} &       &       &       &       &       & \multicolumn{1}{r}{} \\
 	& \multicolumn{1}{l}{     Ten-point scale} & 5.54  & 2.84  &       & 5.48  & 2.89  & = 1, if hard work brings a better life; . . . ; = 10, if it's luck and connections. \\
 	& \multicolumn{1}{l}{     Two-point scale} & 0.48  & 0.50  &       & 0.44  & 0.50  & = 0, if points 1 to 5 in the 10-point scale; = 1, otherwise.  \\
 	\multicolumn{2}{l}{\textbf{Interested in politics}} & 0.90  & 0.30  &       & 0.31  & 0.46  & = 0, if not at all or not very interested; = 1, if somewhat or very interested in politics. \\
 	\multicolumn{2}{l}{\textbf{Strong leader}} & 0.24  & 0.43  &       & 0.39  & 0.49  & = 0, if very or fairly bad; = 1, if fairly or very good. \\
 	\multicolumn{2}{l}{\textbf{Experts rather than government}} & 0.56  & 0.50  &       & 0.48  & 0.50  & = 0, if very or fairly bad; = 1, if fairly or very good. \\
 	\multicolumn{2}{l}{\textbf{Military government}} & 0.04  & 0.20  &       & 0.09  & 0.29  & = 0, if very or fairly bad; = 1, if fairly or very good. \\
 	\multicolumn{2}{l}{\textbf{Democratic government}} & 0.98  & 0.13  &       & 0.95  & 0.22  & = 0, if very or fairly bad; = 1, if fairly or very good. \\
 	\bottomrule
 \end{tabular}%

		\label{table:descriptive1}
		\Fignote{Source: Own computations based on experiment and the World Values Survey.}
	\end{threeparttable}
\end{table}

\begin{table}[htbp]
	\centering
\begin{threeparttable}
	\caption{Descriptive statistics 2: frequency of $SP$.}
    \begin{tabular}{rlcrlc}
	\midrule
	\multicolumn{3}{l}{\textbf{Gender}} & \multicolumn{3}{l}{\textbf{Raise government ownership}} \\
	& Male  & 0.14  &       & 1 to 5 & 0.15 \\
	& Female & 0.18  &       & 6 to 10 & 0.16 \\
	\multicolumn{3}{l}{\textbf{Department}} & \multicolumn{3}{l}{\textbf{People take more responsibility}} \\
	& Montevideo & 0.16  &       & 1 to 5 & 0.17 \\
	& Other & 0.17  &       & 6 to 10 & 0.13 \\
	\multicolumn{3}{l}{\textbf{Father's Education Level }} & \multicolumn{3}{l}{\textbf{Competition is harmful}} \\
	& Non tertiary & 0.17  &       & 1 to 5 & 0.15 \\
	& Tertiary & 0.15  &       & 6 to 10 & 0.19 \\
	\multicolumn{3}{l}{\textbf{Mother's Education Level }} & \multicolumn{3}{l}{\textbf{Luck and contacts}} \\
	& Non tertiary & 0.19  &       & 1 to 5 & 0.17 \\
	& Tertiary & 0.13  &       & 6 to 10 & 0.16 \\
	\multicolumn{3}{l}{\textbf{Primary Education}} & \multicolumn{3}{l}{\textbf{Interest in politics}} \\
	& Private & 0.18  &       & Not at all or not very interested & 0.14 \\
	& Public & 0.15  &       & Somewhat or very interested & 0.17 \\
	\multicolumn{3}{l}{\textbf{High School Education}} & \multicolumn{3}{l}{\textbf{Strong leader}} \\
	& Private & 0.15  &       & Very or fairly bad & 0.16 \\
	& Public & 0.17  &       & Fairly or very good & 0.20 \\
	\multicolumn{3}{l}{\textbf{Income decile (self perception)}} & \multicolumn{3}{l}{\textbf{Experts rather than government}} \\
	& 1 to 5 & 0.17  &       & Very or fairly bad & 0.17 \\
	& 6 to 10 & 0.16  &       & Fairly or very good & 0.16 \\
	\multicolumn{3}{l}{\textbf{Left- to right-wing}} & \multicolumn{3}{l}{\textbf{Military government}} \\
	& 1 to 5 & 0.15  &       & Very or fairly bad & 0.15 \\
	& 6 to 10 & 0.24  &       & Fairly or very good & 0.44 \\
	\multicolumn{3}{l}{\textbf{We need larger income differences}} & \multicolumn{3}{l}{\textbf{Democratic government}} \\
	& 1 to 5 & 0.17  &       & Very or fairly bad & 0.20 \\
	& 6 to 10 & 0.16  &       & Fairly or very good & 0.16 \\
	\bottomrule
\end{tabular}%
		\label{table:descriptive2}%
		\Fignote{Source: Own computations based on experiment and the World Values Survey.}
\end{threeparttable}
\end{table}%

\section{Balance of covariates between treatment arms}\label{balance}

In tables \ref{table:balance1} and \ref{table:balance2}, we present an analysis of balance of some covariates between treatment arms. 

We compare the frequencies of female, subjects aged 19 and less, subjects having attended private primary and secondary school, and subjects whose parents have some university education (even if not necessarily finished).\footnote{We do not have finer individual data because of a concern raised by the ethics committee about the anonymity of subjects.}  

We do not estimate balance for  hypotheses 2 and 4 because these are perfectly balanced by design. Indeed, these hypotheses are tested using only within-subjects comparisons and hence the treated and control groups include exactly the same subjects. The other hypotheses are tested using between- or between- and within-subjects comparisons and hence a balance of covariates can be performed.

\begin{table}[H] 
\footnotesize
		\centering
\begin{threeparttable} 
\caption{Balance testing.}
 \label{table:balance1}
\begin{tabular}{l*{3}{c}}
\hline
 
            &     Treated&     Control&    P-values\\
\cmidrule{2-4}
& \multicolumn{3}{c}{h1}  \\
\cmidrule{2-4}  
Female   &       0.682&       0.631&       0.246\\
Aged 19 and less &       0.191&       0.207&       0.694\\
Private primary school&       0.363&       0.365&       1.000\\
Private high scool&       0.353&       0.362&       0.868\\
Father attended university&       0.389&       0.360&       0.507\\
Mother attended university&       0.465&       0.432&       0.469\\
\cmidrule{2-4}
     & \multicolumn{3}{c}{h3}  \\
\cmidrule{2-4}  
Female   &       0.508&       0.631&       0.054\\
Aged 19 and less &       0.246&       0.207&       0.446\\
Private primary school&       0.369&       0.365&       1.000\\
Private high scool&       0.385&       0.362&       0.701\\
Father attended university&       0.292&       0.360&       0.302\\
Mother attended university&       0.354&       0.432&       0.259\\
\cmidrule{2-4}
     & \multicolumn{3}{c}{h5}
 \\
\cmidrule{2-4}  
Female   &       0.682&       0.508&       0.215\\
Aged 19 and less &       0.191&       0.246&       0.694\\
Private primary school&       0.363&       0.369&       1.000\\
Private high scool&       0.353&       0.385&       0.868\\
Father attended university&       0.389&       0.292&       0.456\\
Mother attended university&       0.465&       0.354&       0.421\\
\cmidrule{2-4}
& \multicolumn{3}{c}{h6}  \\
\cmidrule{2-4}  
Female   &       0.674&       0.619&       0.647\\
Aged 19 and less &       0.239&       0.199&       0.586\\
Private primary school&       0.326&       0.375&       0.648\\
Private high scool&       0.283&       0.383&       0.287\\
Father attended university&       0.196&       0.403&       0.021\\
Mother attended university&       0.413&       0.438&       0.882\\
\hline

\end{tabular}
\Fignote{Columns 1 and 2 contain the proportions of individuals satisfying the condition described in the row name in the treated and control groups. Column 3 contains the P-values for two-tail Fisher tests for the null hypotheses that the proportions are the same in treated and control groups. }
		\end{threeparttable}
	 
\end{table}

\begin{table}[H] 
		\centering

\begin{threeparttable} 
\caption{Balance testing.}
 \label{table:balance2}
\begin{tabular}{l*{3}{c}}
\hline
 
            &     Treated&     Control&    P-values\\
\cmidrule{2-4}
     & \multicolumn{3}{c}{h7}  \\
\cmidrule{2-4}  
Female   &       0.727&       0.614&       0.283\\
Aged 19 and less &       0.152&       0.217&       0.524\\
Private primary school&       0.333&       0.370&       0.857\\
Private high scool&       0.303&       0.372&       0.588\\
Father attended university&       0.455&       0.344&       0.279\\
Mother attended university&       0.333&       0.450&       0.294\\
\cmidrule{2-4}
     & \multicolumn{3}{c}{h8}
 \\
\cmidrule{2-4}  
Female   &       0.679&       0.624&       0.698\\
Aged 19 and less &       0.143&       0.216&       0.492\\
Private primary school&       0.536&       0.340&       0.076\\
Private high scool&       0.500&       0.342&       0.167\\
Father attended university&       0.464&       0.345&       0.325\\
Mother attended university&       0.679&       0.397&       0.012\\
\cmidrule{2-4}
     & \multicolumn{3}{c}{h9}
 \\
\cmidrule{2-4}  
Female   &       0.538&       0.643&       0.417\\
Aged 19 and less &       0.346&       0.189&       0.090\\
Private primary school&       0.385&       0.362&       0.841\\
Private high scool&       0.385&       0.359&       0.840\\
Father attended university&       0.192&       0.383&       0.100\\
Mother attended university&       0.385&       0.439&       0.695\\
\hline

\end{tabular}
\Fignote{Columns 1 and 2 contain the proportions of individuals satisfying the condition described in the row name in the treated and control groups. Column 3 contains the P-values for two-tail Fisher tests for the null hypotheses that the proportions are the same in treated and control groups. }
		\end{threeparttable}
	 
\end{table}
 
Most of the covariates are balanced. We find differences significant at 5 per cent for the proportion of subjects whose father attended the university in testing hypotheses 6 and 8, and the proportion of mothers who attended the university in testing hypotheses 9. None is significant at 1 per cent. We interpret these ``findings'' as the result of chance.

Following recommendations by \textcite{Deaton2018} and, more forcefully, \textcite{Mutz2019}, we did not modify our estimation procedures based on balance testing. \footnote{\textcite{Mutz2019} conclude that ``Our central conclusion is that there is no statistical basis for advocating the Balance Test \& Adjust procedure for analyzing randomized experiments. Although balance testing is widely advocated and is believed to produce more credible estimates of experimental effects, posthoc adjustments using covariates selected on the basis of failed balance tests have no basis in statistical theory. Covariates that are chosen after an experiment is conducted should produce greater rather than lesser skepticism about the results.'' }

\section{Fisher tests}\label{Fisher}

	\begin{table}	[H]
		\centering
		\begin{threeparttable}
			\caption{Frequency of special powers with ex-ante beneficial and harmful reforms.} \label{table:frequencies}
			\centering
			\begin{tabular}{lrccr}
				\toprule
				&       & \multicolumn{2}{c}{\textbf{Reform is ex ante…}} &  \\
				\midrule
				\multicolumn{2}{l}{\textbf{Political gridlock status}} & Beneficial & Harmful &  \\
				\multicolumn{2}{l}{No gridlock} & 0,08  & 0,12  & (0,376) \\
				\multicolumn{2}{l}{Gridlock with …} &       &       &  \\
				& \multicolumn{1}{l}{reformist X and conservative L} & 0,50  & 0,33  & (0,897) \\
				&       & (0,004) & (0,066) &  \\
				& \multicolumn{1}{l}{unbiased X and conservative L} & 0,77  & 0,73  & (0,685) \\
				&       & (0,000) & (0,000) &  \\
				& \multicolumn{1}{l}{reformist X and unbiased L} & 0,20  & 0,27  &  (0,366)\\
				&       & (0,176) & (0,140) &  \\
				\multicolumn{2}{l}{Number of subjects} & 40    & 49    &  \\
				Number of observations &       & 80    & 98    &  \\
				\bottomrule
			\end{tabular}%
			
			\Fignote{Notes: Data from treatments 3 and 4, corresponding to ex ante beneficial and harmful reforms, respectively. Under the heading ``No gridlock'', we report results obtained with conservative X and L. The heading ``Gridlock with reformist X and conservative L'' is self explanatory. Under the heading ``Gridlock with unbiased X and conservative L'', we report results with a unbiased X proposing reform and a conservative L. Under the heading 
				``Gridlock with reformist X and unbiased L'', we report results with a reformist X and a unbiased L proposing the status quo policy. Below each frequency with political gridlock, we report the one-sided p-value of a Fisher test of difference in proportions in which the null is that the observed frequency is equal to the baseline presented in the no gridlock case and the alternative is that the frequency with gridlock is higher or lower than the baseline, depending on the model's prediction in each condition. In the third column, we report the one-sided  p-values of Fisher tests of differences of proportions in which the null is that the probability of $SP$ is lower than or equal to in the ex ante harmful than in the beneficial reform environment and the alternative is that it is higher in the ex ante harmful reform environment. 			 \\
				Source: Own computations based on experimental data.}
		\end{threeparttable}
	\end{table}

	\begin{table}[H]
		\centering
		\begin{threeparttable}
		\caption{Frequency of special powers with low and high rents}    \label{table:rents}		
		\begin{tabular}{lrccr}
			\toprule
			&       & \multicolumn{2}{c}{\textbf{Rents are …}} &  \\
			\midrule
			\multicolumn{2}{l}{\textbf{Political gridlock status}} & Low   & High  & \multicolumn{1}{c}{p-value} \\
			\multicolumn{2}{l}{No gridlock} & 0,08  & 0,09  & \multicolumn{1}{c}{0,681} \\
			\multicolumn{2}{l}{Gridlock with …} &       &       &  \\
			& \multicolumn{1}{l}{reformist X and conservative L} & 0,50  & 0,13  & \multicolumn{1}{c}{0,006} \\
			& \multicolumn{1}{l}{unbiased X and conservative L} & 0,77  & 0,24  & \multicolumn{1}{c}{0,000} \\
			& \multicolumn{1}{l}{reformist X and unbiased L} & 0,20  & 0,09  & \multicolumn{1}{c}{0,212} \\
			\multicolumn{2}{l}{Number of subjects} & 40    & 46    &  \\
			\bottomrule
		\end{tabular}%
		
		\Fignote{Notes: Data from treatments 3 and 5, corresponding to low and high rents, respectively. Under the heading ``No gridlock'', we report results obtained with conservative X and L. The heading ``Gridlock with reformist X and conservative L'' is self explanatory. Under the heading ``Gridlock with unbiased X and conservative L'', we report results with a unbiased X proposing reform and a conservative L. Under the heading ``Gridlock with reformist X and unbiased L'', we report results with a reformist X and a unbiased L proposing the status quo policy. In the third column we report the one-sided p-value of a Fisher test of difference in proportions in which the null is that the observed frequency is the same regardless of rents and the alternative is that the frequency is lower with high than low rents.  			 \\
		Source: Own computations based on experimental data.}
		\end{threeparttable}
	\end{table}%

\begin{table}[H]
	\centering
	\begin{threeparttable}	
		\caption{Frequency of special powers: the impact of framing}    \label{table:framing}
		\begin{tabular}{lrccrcr}
			\toprule
			&       & \multicolumn{5}{c}{\textbf{Framing is …}} \\
			\cmidrule{3-7}          &       & Neutral & Corruption & \multicolumn{1}{c}{p-value} & Political & \multicolumn{1}{c}{p-value} \\
			&       & (1)   & (2)   & \multicolumn{1}{c}{(3)} & (4)   & \multicolumn{1}{c}{(5)} \\
			\multicolumn{2}{r}{} &       &       & \multicolumn{1}{c}{H0: (1) = (2)} &       & \multicolumn{1}{c}{H0: (2) = (4)} \\
			\cmidrule{3-7}    \multicolumn{2}{l}{No gridlock} & 0,08  & 0,00  & \multicolumn{1}{c}{0,637} & 0,03  & \multicolumn{1}{c}{0,42} \\
			\multicolumn{2}{l}{Gridlock with …} &       &       &       &       &  \\
			& \multicolumn{1}{l}{reformist X and conservative L} & 0,50  & 0,48  & \multicolumn{1}{c}{1,000} & 0,48  & \multicolumn{1}{c}{1,00} \\
			& \multicolumn{1}{l}{unbiased X and conservative L} & 0,77  & 0,45  &  \multicolumn{1}{c}{0,799} & 0,55  & \multicolumn{1}{c}{0,73} \\
			& \multicolumn{1}{l}{reformist X and unbiased L} & 0,20  & 0,30  & \multicolumn{1}{c}{0,396} & 0,15  & \multicolumn{1}{c}{0,41} \\
			\multicolumn{2}{l}{Number of subjects} & 40    & 33    &       & 33    &  \\
			\bottomrule
		\end{tabular}%
		
		\Fignote{Notes: Data from treatments 3, 6 and 7, corresponding to neutral, corruption (i.e. neutral plus cost of rule 2 identified as corruption) and political framing, respectively. Under the heading ``No gridlock'', we report results obtained with conservative X and L. The heading ``Gridlock with reformist X and conservative L'' is self explanatory. Under the heading ``Gridlock with unbiased X and conservative L'', we report results with a unbiased X proposing reform and a conservative L. Under the heading ``Gridlock with reformist X and unbiased L'', we report results with a reformist X and a unbiased L proposing the status quo policy. In columns $(3)$ and $(5)$ we report the two-sided p-value of a Fisher test of difference in proportions in which the null is that the observed frequency is the same regardless of framing and the alternative is that the frequency is different.  			 \\
			Source: Own computations based on experimental data.}
		
	\end{threeparttable}
\end{table}%
		
		\begin{table}[H]
		\centering
		\begin{threeparttable}	
			\caption{Frequency of special powers: the impact of the frequency of gridlock in the priming phase}    \label{table:lowfreqpriming}
			\begin{tabular}{lrccr}
				\toprule
				&       & \multicolumn{2}{c}{\textbf{Frequency of gridlock in priming phase …}} &  \\
				\midrule
				\multicolumn{2}{l}{\textbf{Political gridlock status}} & Low   & High  & \multicolumn{1}{c}{p-value} \\
				\multicolumn{2}{l}{No gridlock} & 0,12  & 0,08  & \multicolumn{1}{c}{0,793} \\
				\multicolumn{2}{l}{Gridlock with …} &       &       &  \\
				& \multicolumn{1}{l}{reformist X and conservative L} & 0,69  & 0,50  & \multicolumn{1}{c}{0,922} \\
				& \multicolumn{1}{l}{unbiased X and conservative L} & 0,77  & 0,77  & \multicolumn{1}{c}{0,559} \\
				& \multicolumn{1}{l}{reformist X and unbiased L} & 0,31  & 0,20  & \multicolumn{1}{c}{0,843} \\
				\multicolumn{2}{l}{Number of subjects} & 26    & 40    &  \\
				\bottomrule
			\end{tabular}%
			\Fignote{Notes: Data from treatments 1 and 3, corresponding to low and high frequency of political gridlock in the priming phase, respectively. Under the heading ``No gridlock'', we report results obtained with conservative X and L. The heading ``Gridlock with reformist X and conservative L'' is self explanatory. Under the heading ``Gridlock with unbiased X and conservative L'', we report results with a unbiased X proposing reform and a conservative L. Under the heading ``Gridlock with reformist X and unbiased L'', we report results with a reformist X and a unbiased L proposing the status quo policy. In the third column we report the one-sided p-value of a Fisher test of difference in proportions in which the null is that the observed frequency of $SP$ is the same regardless of priming and the alternative is that the frequency of $SP$ is lower when political gridlock was less frequent in the priming phase.  			 \\
				Source: Own computations based on experimental data.}
			
		\end{threeparttable}
	\end{table}%

\newpage	
\section{Instructions of treatments 1 and 7 \normalfont{(translated from the original version in Spanish).}}\label{Instr}

\includegraphics[scale=0.92]{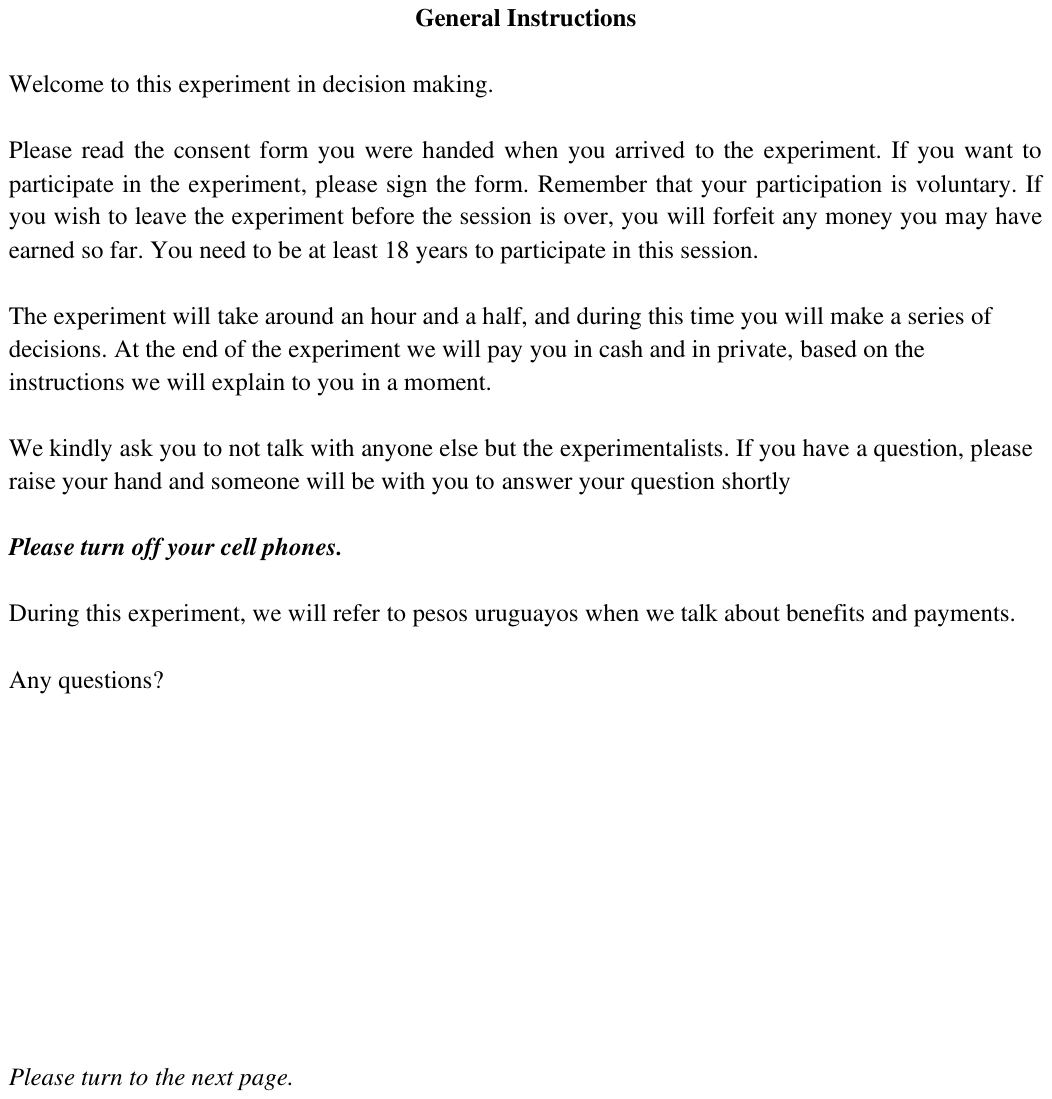}

\includepdf[pages=-]{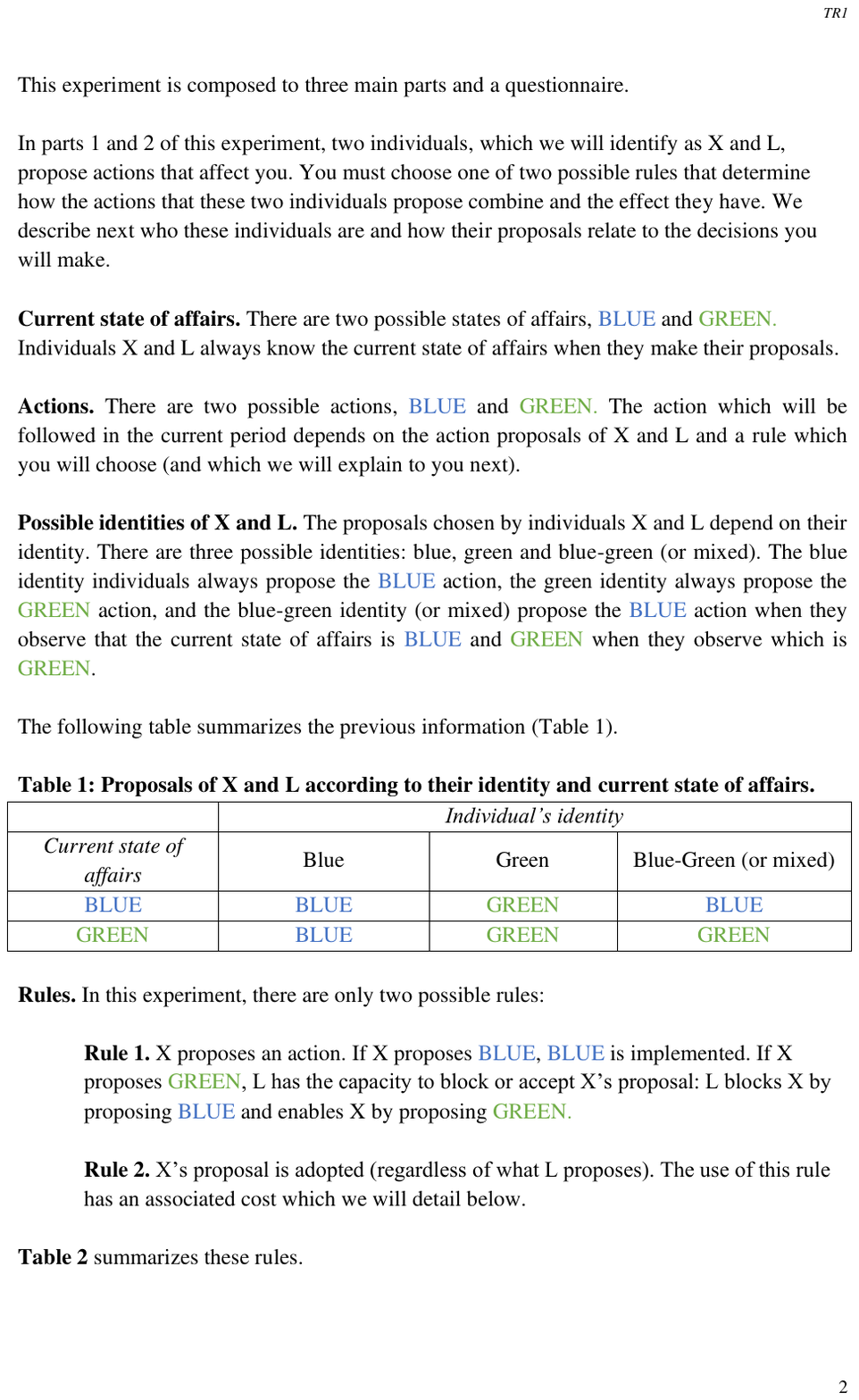}

\includepdf[pages=-]{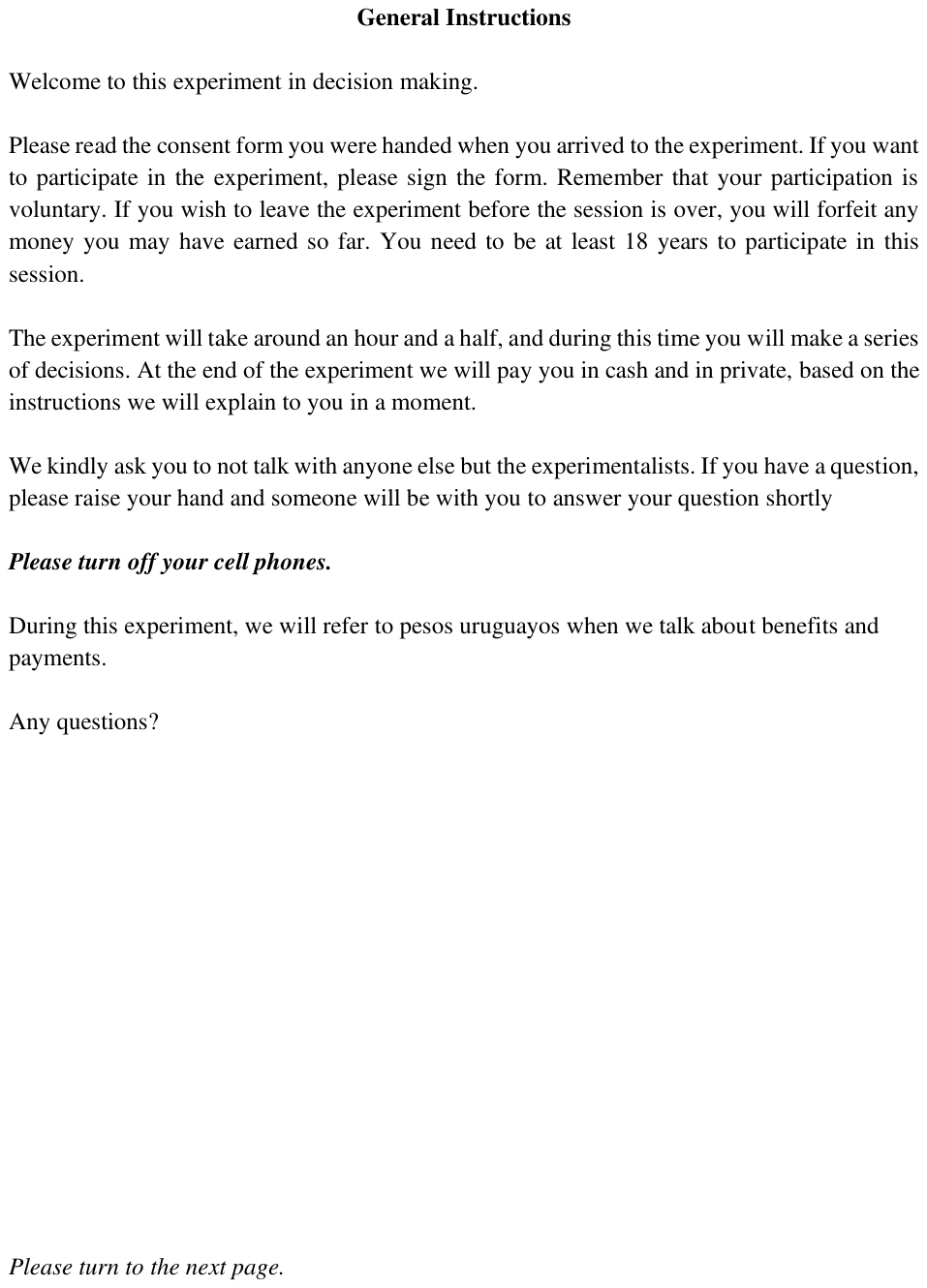}

\section{Risk aversion measurement.}\label{Risk}	
\includegraphics[scale=0.65]{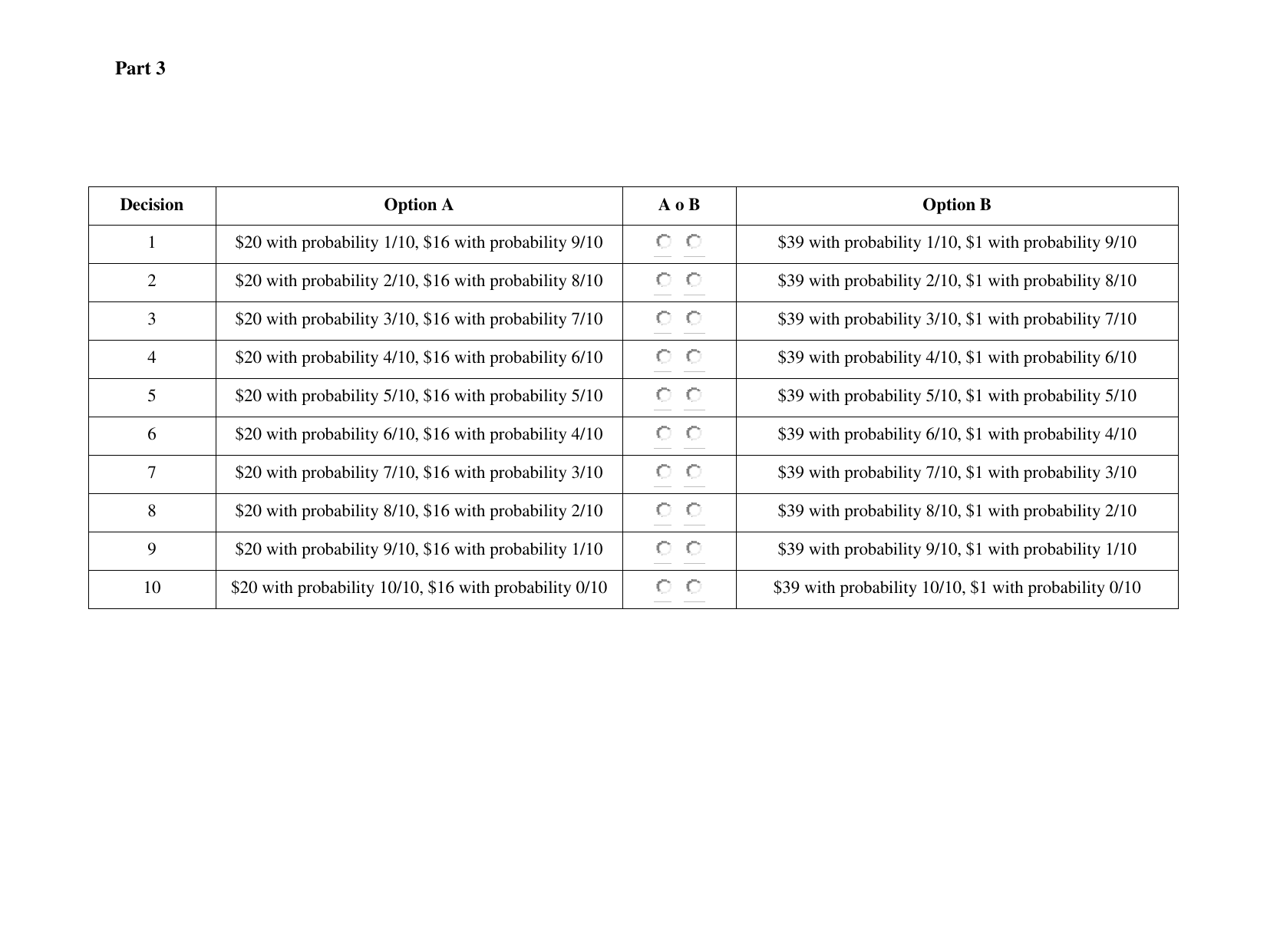}
\newpage

\section{Post-experimental questionnaire}\label{Ques}	
\includepdf[pages=-]{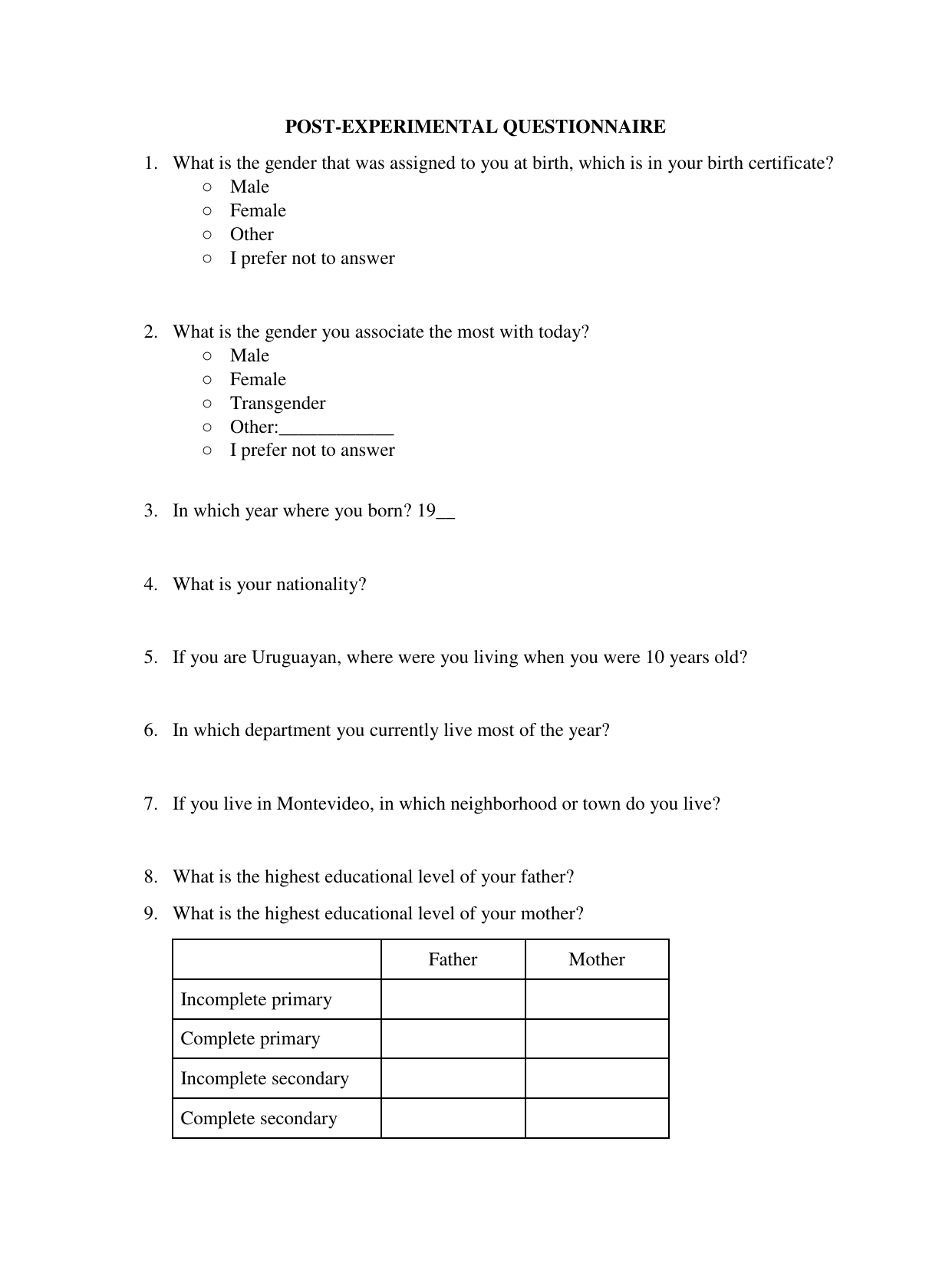}

\end{document}